\documentclass{amsart}
\usepackage{graphicx} 
\usepackage{amsthm}
\usepackage{amssymb}
\usepackage{graphicx}
\usepackage{enumerate}
\usepackage{xcolor}
\usepackage{verbatim}
\usepackage[hyperindex,breaklinks]{hyperref}
\usepackage{braket}
\usepackage[all,cmtip]{xy}
\usepackage{tikz}
\usepackage{tikz-cd}
\usetikzlibrary{positioning, quotes, decorations}
    \tikzset{
every loop/.style={},
path/.style={line width=2.2pt, red},
tadpole/.style={loop, min distance=15mm, in=-40, out=40},
bdry/.style={shape=circle, draw, inner sep=2pt,fill=black!40},
bulk/.style={shape=circle, draw, inner sep=2pt,fill=black},
bubu/.style={black,ultra thick},
bubu1/.style={black},
bubo/.style={decorate,decoration=snake,red},
bobo/.style={decorate,decoration={coil,amplitude=0.1cm, aspect=1.2,segment length=0.2cm},purple},
bobo2/.style={densely dashed, decorate, decoration={snake,amplitude=0.1cm},blue},
bobo3/.style={decorate, decoration={zigzag,amplitude=0.1cm},green}
}
    \usetikzlibrary{calc}
\usetikzlibrary{decorations.pathmorphing,shapes}

\usepackage{subcaption}
\usepackage{float}  
\usepackage{adjustbox} 

\allowdisplaybreaks
                
\newcommand{\RR}{\mathbb{R}}
\newcommand{\mr}{\mathrm}
\newcommand{\ra}{\rightarrow}

\newcommand{\CC}{\mathbb{C}}

\newcommand{\dd}{\partial}
\newcommand{\DN}{\mathrm{DN}}
\newcommand{\wh}{\widehat}
\newcommand{\til}{\widetilde}
\newcommand{\val}{\mathrm{val}}

\newcommand{\pert}{\mathrm{pert}}

\newcommand{\ZZ}{\mathbb{Z}}
\newcommand{\red}{\color{red}}
\newcommand{\green}{\color{green}}
\newcommand{\Kp}{K'}
\newcommand{\gp}{\widetilde{\gamma}}
\newcommand{\Mp}{\widetilde{M}}
\newcommand{\cyl}{\mathrm{cyl}}
\newcommand{\cylp}{\widetilde{\mathrm{cyl}}}
\newcommand{\gpp}{\widetilde{\widetilde{\gamma}}}
\newcommand{\Mpp}{\widetilde{\widetilde{M}}}
\newcommand{\cylbar}{\overline{\mathrm{cyl}}}
\newcommand{\ep}{\tilde{\epsilon}}
\newcommand{\DD}{\mathbb{D}}
\newcommand{\re}{\operatorname{Re}}
\newcommand{\ogam}{\overline{\gamma}}
\newcommand{\ugam}{\underline{\gamma}}
\newcommand{\ougam}{\overline{\underline{\gamma}}}

\theoremstyle{remark}
\newtheorem{remark}{Remark}[section]
\theoremstyle{plain}
\newtheorem{lemma}[remark]{Lemma}
\newtheorem{proposition}[remark]{Proposition}
\newtheorem{thm}[remark]{Theorem}
\newtheorem{corollary}[remark]{Corollary}
\theoremstyle{definition}

\newtheorem{example}[remark]{Example}
\newtheorem{assumption}{Assumption}

\newtheorem{conjecture}[remark]{Conjecture}

\newcommand{\EE}{\mathcal{E}}
\newcommand{\LL}{\mathbf{L}}
\newcommand{\ii}{\mathrm{i}}
\title{Gluing formulae for heat kernels}

\begin{document}
\author[P. Mnev]{Pavel Mnev}
\address{University of Notre Dame, Notre Dame, IN 46556, USA}
\email{pmnev@nd.edu}

\author[K. Wernli]{Konstantin Wernli}
\address{Centre for Quantum Mathematics, IMADA, University of Southern Denmark, Campusvej 55, 5230 Odense M, Denmark}
\email{kwernli@imada.sdu.dk}

\thanks{
This work is partly a result of the ERCSyG project, Recursive and Exact New Quantum Theory (ReNewQuantum) which received
funding from the European Research Council (ERC) under the European Union’s Horizon 2020 research and innovation programme
under grant agreement No. 81057, and supported K. W. during the completion of this work.}
\begin{abstract}
We state and prove two gluing formulae for the heat kernel of the Laplacian on a Riemannian manifold of the form $M_1 \cup_\gamma M_2$. We present several examples.
\end{abstract}
\maketitle

\setcounter{tocdepth}{3}
\tableofcontents

\section{Introduction} 
The goal of this paper is to investigate the behavior of the heat kernel of the Laplace-Beltrami operator on a Riemannian manifold $M$ decomposed along a 
codimension one submanifold $\gamma \subset M$ into two Riemannian manifolds with boundary $M_1$ and $M_2$, i.e., $M = M_1 \cup_\gamma M_2$. 
Namely, we obtain a formula (formula (\ref{gluing formula 2, intro}) below) for the heat kernel on $M$, expressed purely in terms of the heat kernels on $M_1, M_2$, and $\gamma$. Our proof requires a certain assumption (Assumption \ref{conj: bound} on p. \pageref{conj: bound}) on the Dirichlet-to-Neumann operators of the manifolds $M_1,M_2$ and relies on an intermediate gluing formula (\ref{gluing 1 intro}) that is a consequence of a gluing formula for Green's functions. 

Our original motivation comes from perturbative quantum field theory, see Appendix \ref{sec: motivation}. 
Given the ubiquitousness of the heat kernel, this formula is expected to 
have other applications. We envision that it could be useful, e.g., in the connection of the heat kernel to random walks, or lead to another proof of the Burghelea-Friedlander-Kappeler gluing formula for determinants \cite{BFK}. 

\subsection{Main Result}\label{sec: Results}
Here is the main theorem of the paper (summarizing Proposition \ref{prop: gluing heat kernel 1} and Corollary \ref{cor: gluing formula 2} in the main text). 
\begin{thm}\label{thm: intro}
    Let $M$ be a Riemannian manifold, possibly with boundary $\dd M$, split by a 
    {compact} codimension one submanifold $\gamma$ disjoint from $\dd M$ into $M_1$ and $M_2$: $M=M_1\cup_\gamma M_2$.
Then:

\begin{enumerate}[(i)]
    \item The heat kernel on $M$ can be expressed in terms of the data associated to $M_1$ and $M_2$ as follows. For $x\in M_a$, $y\in M_b$ with $a,b\in\{1,2\}$, one has
\begin{multline}\label{gluing 1 intro}
    K(x,y|t) = \delta_{ab}K^{(a)}(x,y|t) +  \\
    \int_{t_0+t_1+t_2 = t, t_i > 0}dt_0dt_1\, \int_{\gamma \times \gamma} dudv\, \partial^n_uK^{(a)}(x,u|t_0)\LL^{-1}[\DD^{-1}_{m^2}(u,v)](t_1)\partial^n_vK^{(b)}(v,y|t_2).
\end{multline}
Here the notations are:
\begin{itemize}
    \item $K$ and $K^{(1,2)}$ are the heat kernels on $M$ and $M_{1,2}$ with Dirichlet boundary conditions.
    \item $\dd^n$ is the normal derivative on the interface $\gamma$.
    \item $\DD_{m^2}(u,v)=\DN^{\gamma,M_1}_{m^2}+\DN^{\gamma,M_2}_{m^2}$ is the sum of two Dirichlet-to-Neumann operators, corresponding to extending a function $\eta\in C^\infty(\gamma)$ to a solution of the Helmholtz equation $(\Delta+m^2)\phi=0$, $\phi|_\gamma=\eta$ on $M_{1,2}$ and taking the normal derivative of $\phi$ on $\gamma$.
    \item $\LL^{-1}$ is the inverse Laplace transform of a function of $m^2$ to a function of $t$.
\end{itemize}
\item Moreover, under an extra assumption on $M_{1,2}$ (Assumption \ref{conj: bound}; in particular, it holds for 
$M$ having product metric near $\gamma$), one can directly express the heat kernel on $M$ in terms of heat kernels on $M_{1,2}$ and on the interface $\gamma$:
\begin{multline}\label{gluing formula 2, intro}
    K(x,y|t) = \delta_{ab} K^{(a)}(x,y|t) + \\
    +\sum_{k\geq 0}(-1)^k\int_{\sum_{i=0}^{2k+2} t_i = t, t_i> 0}\prod_{i=0}^{2k+1}dt_i\int_{\gamma^{2k+2}}\prod_{i=0}^{2k+1} du_i\,
    \partial^n_{u_0}K^{(a)}(x,u_0|t_0) \cdot \\ 
    \cdot \prod_{i=0}^{k-1}\frac{1}{\sqrt{4\pi t_{2i+1}}}K^{\gamma}(u_{2i},u_{2i+1}|t_{2i+1})\Kp(u_{2i+1},u_{2i+2}|t_{2i+2})\cdot \\
    \frac{1}{\sqrt{4\pi t_{2k+1}}}\cdot K^{\gamma}(u_{2k},u_{2k+1}|t_{2k+1})\partial^n_{u_{2k+1}}K^{(b)}(u_{2k+1},y|t_{2k+2}).
\end{multline} 
Here we use the notation
\begin{equation*}
  \Kp(u,v|t) =  \partial^n_u\partial^n_v \left(K^{(1)}(u,v|t) + K^{(2)}(u,v|t) \right)- \frac{1}{\sqrt{\pi} t^{3/2}}K^{\gamma }(u,v|t).
\end{equation*}
\end{enumerate}
\end{thm}
\subsection{Plan of the paper}
Overall we have made an attempt at keeping the main body of the paper concise. Technical proofs and digressions have been moved to appendices. 

In Section \ref{sec: preliminaries} we collect preliminaries on the Laplace transform and its application to the relation between the heat kernel and the Green's function. We also recall some facts about the Dirichlet-to-Neumann operator.  

In Section \ref{sec: conj} we formulate a technical assumption about the Dirichlet-to-Neumann operator that is required for the proof of our main theorem. We expect this assumption to be satisfied on any Riemannian manifold with boundary (Conjecture \ref{conjecture}). 

Section \ref{sec: gluing formula} is the main part of the paper and contains our main results, two gluing formulae for the heat kernel on a Riemannian manifold $M$ with a decomposition $M = M_1 \cup_\gamma M_2$. The first one is Proposition \ref{prop: gluing heat kernel 1}, the second one is Theorem \ref{thm: gluing formula} and Corollary \ref{cor: gluing formula 2}. 

As an aside, in Subsection \ref{ss: cutting formulae} we present a related pair of ``cutting-out'' formulae for heat kernels, allowing one to recover the heat kernel on a closed domain in $M$ with Dirichlet boundary condition from the heat kernel on $M$.

Section \ref{sec: examples} contains some examples of the gluing formulae. 

In Appendix \ref{sec: motivation} we explain our original motivation for writing this paper, the study of renormalization in 
quantum field theory and how it interacts with cutting and gluing of the spacetime manifold. 

In Appendices \ref{app: Proof of Prop} and \ref{appendix: proof of Lemma} contain the proofs of technical results. 

Appendix \ref{appendix: DN examples for Conj} contains some examples of Conjecture \ref{conjecture}. 

In Appendix \ref{app: path integral} we discuss the interpretation of our gluing formula in terms of the path integral representation of the heat kernel (Feynman-Kac formula) at a heuristic level. 

Finally, Appendix \ref{sec: discrete} contains a gluing formula for the heat kernel on a graph, and its interpretation in terms of sums over paths on a graph. This is a rigorous discrete analog of the heuristic discussion in the previous appendix.

\subsection{Conventions}
In general, we work on arbitrary Riemannian manifolds, possibly with boundary. However, for most of our results we have to assume that the boundary (or at least, the gluing interface) is compact (although we present examples where the gluing interface is not compact).

Our convention is that the Laplace-Beltrami operator $\Delta^M \colon C_c^\infty(M) \to C_c^\infty(M)$ is a non-negative (densely defined)  operator on $L^2(M)$.

For a general Riemannian manifold, we work with the self-adjoint Friedrichs extension of the Laplace-Beltrami operator $\Delta^M$ and its associated heat kernel (see Section \ref{sec: arb mfld}). 
For manifolds with boundary, Dirichlet boundary conditions are assumed by default.

For $M$ a Riemannian manifold,  $\int_M dx \cdots$ denotes the integral with respect to Riemannian volume form on $M$.

Operators are denoted with superscripts according to which space they act on, e.g. $\Delta^M$ denotes the Laplace-Beltrami operator on $M$, etc. When there is no danger of confusion, the superscript will sometimes be dropped. If the operator depends on extra parameters, those will be denoted with a subscript, and again this might be dropped when there is no danger of confusion. 

\subsection*{Acknowledgements}
We thank Olga Chekeres, Ivan Contreras, Santosh Kandel, Andrey Losev, Donald R. Youmans for helpful discussions related to the gluing formula for Green's functions. We also thank the anonymous referee for helpful remarks.

\section{Preliminaries}\label{sec: preliminaries}
\subsection{Laplace transform}
The Laplace transform of a locally integrable function $f(t)$ on $[0,\infty)$
is the function $F(s)$ of the complex parameter $s$ defined by\footnote{Slightly more generally, one defines the Laplace transform of a measure $\mu$ on $[0,\infty)$ to be $$\int_{\RR_{\geq 0}}e^{-st}d\mu.$$
Below we will sometimes use the point measure at 0 and denote it $\delta(t)dt$,  its Laplace transform is the constant function $F(s) \equiv 1$. }
\begin{equation*}
   F(s) = \LL[f](s) = \int_0^\infty f(t)e^{-st}dt.
\end{equation*}
The function $F(s)$ is defined on the set $R_f =\{s \in \CC \,|\, \int_0^\infty |f(t)e^{-st}|dt < \infty\}$, which can be seen to be of the form {
$\operatorname{Re}(s) > a_f$ or $\operatorname{Re}(s) \geq a_f$},
for some $a_f \in [-\infty,\infty]$. The set $R_f$ is called the region of (absolute) convergence and $a_f$ is known as the abscissa of absolute convergence. 

Conversely, if $F(s) $ is a holomorphic function of a complex parameter $s$, defined on a set of the form $\operatorname{Re} s > a$, the inverse Laplace transform of $F$ is the function $f(t)$ defined by
\begin{equation*}
    f(t) = \LL^{-1}[F](t) = \frac{1}{2\pi \ii}\int_{\gamma - \ii\infty}^{\gamma + \ii\infty} F(s)e^{st}ds
\end{equation*}
with $\gamma >  a$. It is a classical fact that if $F(s) = \LL[f](s)$ and $\gamma > a_f$, then $\LL^{-1}[F](t) = f(t)$.
An important property of the inverse Laplace transform is that it maps product to convolution, 
\begin{equation}
    \LL^{-1}\left[\LL[f](s)\cdot \LL[g](s)\right](t) = \int_0^tf(\tau)g(t-\tau) d\tau. \label{eq: prod to conv} 
\end{equation}

\subsection{The Helmholtz operator and its square root 
{on closed manifolds}}\label{ss: helmholtz}
Let $M$ be a closed\footnote{In this subsection manifolds are required to be compact. In the rest of the paper we are not making such an assumption, unless stated otherwise. We will not assume that manifolds are connected.} Riemannian manifold and $\Delta^M \colon C^\infty(M) \to C^\infty(M)$ the non-negative Laplace operator on $C^\infty(M)$, 
 \begin{equation*}
     \Delta^M \equiv \Delta = d^* d = -*d*d.
 \end{equation*}

 It is a non-negative 
 {densely defined} operator on $L^2(M)$ with spectrum 
 \begin{equation}\label{omegas}
     0 = \omega_1 \leq \omega_2 \leq \omega_3 \leq \ldots
 \end{equation}
with $\omega_k \to \infty$ as $k \to \infty$. 

For any $m >0$, the Helmholtz operator $\Delta + m^2 \colon C^\infty(M) \to C^\infty(M)$ is a (strictly) positive, densely defined operator on $L^2(M)$. It admits a unique self-adjoint extension, and hence has a unique positive square root $\sqrt{\Delta+m^2} \colon C^\infty(M) \to C^\infty(M)$. 

For our purposes it is important to consider the Helmholtz operator where the mass is allowed to have a nonzero imaginary part. For $m \in \CC$ with $\operatorname{Re} m > 0$, the Helmholtz operator is no longer positive, but there still is a unique square root $\sqrt{\Delta + m^2}$ whose eigenvalues $\mu_k = \sqrt{\omega_k + m^2}$ have positive real part and we have 
\begin{equation}\label{mus}
    0 < \operatorname{Re} \mu_1 = \operatorname{Re} m \leq \operatorname{Re} \mu_2 \leq \ldots 
\end{equation}
with $\operatorname{Re} \mu_k \to \infty$ as $k \to \infty$.\footnote{In other words, we choose the unique square root of $\Delta + m^2$ whose real part is a positive operator.} 

The Helmholtz operator is invertible and we denote by $G_{m^2}(x,y)$ the integral kernel of its inverse, called the Green's function. 

One can also allow $M$ to have 
boundary and consider the spectrum of $\Delta$ and $\sqrt{\Delta+m^2}$ with Dirichlet boundary condition. In this case, inequalities (\ref{omegas}) should be adjusted to $0<\omega_1\leq \omega_2\leq \cdots$ and (\ref{mus}) -- to $0<\mr{Re}\,m<\mr{Re}\,\mu_1\leq \mr{Re}\,\mu_2\leq \cdots$.
{
\subsection{Dirichlet heat kernels on arbitrary manifolds}\label{sec: arb mfld}
Let $M$ be an arbitrary manifold, for now without boundary. The Laplacian $\Delta\colon C_c^\infty(M) \to C_c^\infty(M)$ generally has various self-adjoint extension, corresponding to different boundary (or fall-off) conditions. Those correspond to various heat semigroups. However, there is always a canonical extensions, called the Dirichlet or Friedrichs extension, defined as follows: For $\xi,\eta \in C_c^\infty(M)$, define the bilinear form 
$$(\xi, \eta) \mapsto \langle \xi,\eta\rangle_F = \langle \xi, \Delta \eta\rangle_{L^2} + \langle \xi,\eta\rangle_{L^2} = \langle \nabla\xi, \nabla \eta\rangle_{L^2} + \langle \xi,\eta\rangle_{L^2}.$$ 
The closure of $C_c^\infty(M) \subset L^2(M)$ with respect to $\langle\xi,\eta\rangle_F$ is the Sobolev space $H^{1,2}_0(M)$. Next, consider the subspace $H \subset H^{1,2}_0(M)$ consisting of those $\xi$ such that $\eta \mapsto \langle \eta,\xi \rangle_F$ is bounded in the $L^2$ norm, 
by the Riesz representation theorem, there is a vector $A\xi$ such that $\langle \eta,\xi\rangle_F = \langle \eta,A\xi\rangle_{L^2}$. Then one can show that $A\colon H \to L^2(M)$ is a self-adjoint operator and $A-I$ is a self-adjoint extension of the Laplacian. We will denote this extension by the same symbol $\Delta$.\footnote{This extension is the unique self-adjoint one - i.e. the Laplacian is essentially self-adjoint - if the manifold is complete, see \cite{Strichartz}.} Associated to it is a unique heat kernel $K(x,y|t)$, 
i.e. the kernel of heat semigroup $e^{t\Delta}$, see e.g. \cite{Dodziuk} or \cite{GrigoryanNotes}. This heat kernel is constructed by letting $$K(x,y|t) = \sup_{\Omega}K_{\Omega,D}(x,y|t),$$ where $\Omega$ ranges over precompact sets with smooth boundary in $M$ and $K_{\Omega,D}$ denotes the heat kernel on $\Omega$ with Dirichlet boundary conditions.\footnote{One has the following monotonicity property with respect to domains: If $\Omega\subset\Omega'\subset M$ and $x,y \in \Omega$, then $K_{\Omega,D}(x,y|t) \leq K_{\Omega',D}(x,y|t)$, see e.g. \cite[Exercise 7.40]{GrigoryanBook} .}
In much the same way, if $M$ is an arbitrary manifold with arbitrary boundary $\gamma$, one can construct a heat kernel with Dirichlet boundary conditions on $\gamma$, we will also denote this heat kernel by $K$. }

\subsection{A duality between the heat kernel and the Green's function for the Helmholtz operator} 
The proof of our gluing formula rests on the observation that the Laplace transform interchanges the heat kernel of the Laplacian, considered as function of $t\in [0,\infty)$, and the Green's function $ G_{m^2}$ of the Helmholtz operator, considered as a holomorphic function in $m^2$.  Namely, the heat kernel for the Helmholtz operator $\Delta + m^2$ is 
\begin{equation*}
    K_{m^2}(x,y|t) = e^{-m^2t}K(x,y|t). 
\end{equation*}
For $x \neq y$ fixed, the Green's function of $\Delta + m^2$ is given by
\begin{equation}
    G_{m^2}(x,y) = \int_0^\infty K_{m^2}(x,y|t) dt = \int_0^\infty e^{-m^2t}K(x,y|t) dt = \LL[K(x,y)](m^2). \label{eq:Greens and Heat}
\end{equation} 
The convergence of this integral is uniform on compact subsets of $(M \times M) \setminus \Delta$, as the heat kernel $K_{m^2}$ vanishes exponentially for $t\to 0$ and $t \to \infty$ on such subsets.
Therefore, the heat kernel is the inverse Laplace transform of the Green's function 
\begin{equation*}
    K(x,y|t) = \LL^{-1}[G_{m^2}(x,y)](t). 
\end{equation*}
This relation is a reformulation of the well-known relation between the resolvent of a self-adjoint operator and its heat kernel via a contour integral, see for instance \cite[Eq. (10.21)]{Shubin} or \cite[p.48]{Gilkey}.
 The assertions in this subsection, in contrast to those on the spectrum of $\Delta^M$ in Section \ref{ss: helmholtz}, are valid also in the case where $M$ is not compact. The simplest example is $M = \RR$, where the Laplacian (on $L^2(\RR)$) has continuous spectrum $\sigma(\Delta_\RR) = [0,\infty)$.
\begin{example}[Heat kernel on $\RR$]
    The heat kernel on $\RR$ is 
    \begin{equation*}
        K_{m^2}(x,y|t) = \frac{1}{(4\pi t)^{1/2}}e^{-m^2 t- (x-y)^2/(4t)}.
    \end{equation*}
   Integrating over $t$ we obtain\footnote{In one dimension, the singularity of the heat kernel trace at $t= 0$ is integrable, so that the integral actually also converges for $x = y$. }
    \begin{equation*}
        \LL[K_0(x,y)](m^2) = \int_0^\infty \frac{1}{(4\pi t)^{1/2}}e^{-m^2 t- (x-y)^2/(4t)} dt = \frac{e^{-m|x-y|}}{2m} = G_{m^2}(x,y).
    \end{equation*}
    Conversely, denoting $s = m^2$ we get for any $\gamma > 0$
\begin{equation*}
    \LL^{-1}[G_{m^2}(x,y)] = \frac{1}{2\pi\ii}\int_{\gamma -  \ii\infty}^{\gamma +\ii\infty}\frac{e^{ -\sqrt{s}|x-y|}}{2\sqrt{s}}e^{st} ds = \frac{1}{(4\pi t)^{1/2}}e^{-(x-y)^2/(4t)} = K(x,y|t).
\end{equation*}
Evaluating the integral over $s$ is a nontrivial exercise in using a keyhole contour.
\end{example}
\begin{example}[Euclidean space]
For general $n$, note that for $m,a > 0$ we can evaluate the integral 
\begin{equation}
\begin{aligned}
\int_{0}^\infty t^{-n/2}e^{-m^2t - a^2/t} dt &= \int_{-0}^\infty t^{-n/2+1}e^{-ma(\frac{mt}{a}+\frac{a}{mt})}\frac{dt}{t} \\&=  \left(\frac{m}{a}\right)^{n/2-1}\int_{-\infty}^\infty e^{s(-n/2 +1)}e^{-2ma\cosh s} ds \\
&=  2\left(\frac{m}{a}\right)^{n/2-1}\int_0^\infty \cosh((-n/2+1)s)\, e^{-2ma\cosh s}ds\\ &=  2\left(\frac{m}{a}\right)^{n/2-1}K_{-n/2+1}(2ma),
\end{aligned}
\end{equation}
where we have substituted $ s = \log\frac{mt}{a}$, and $K_\alpha(x) = \int_0^\infty e^{-x\cosh t}\cosh \alpha t\  dt$ denotes the modified Bessel function of the second kind.\footnote{This is the solution of $x^2y''(x) + xy'(x) - (x^2 + a^2)y(x) = 0$ that diverges as $x \to 0$ and vanishes as $x \to \infty$. } 
Applying this to the heat kernel on $\RR^n$
$$K_{m^2}(x,y|t) = \left(\frac{1}{4\pi t}\right)^{n/2} e^{-m^2 t - |x-y|^2/t}$$ 
we obtain the $n$-dimensional Green's function of the Helmholtz equation 
\begin{equation}G_{m^2}(x,y) = \frac{1}{2\pi}\left(\frac{m}{4\pi}\right)^{n/2 -1}\frac{K_{1-n/2}(2m|x-y|)}{|x-y|^{n/2 -1}} \label{eq: Green's function Helmholtz R^n}.
\end{equation}
\end{example}
\begin{example}[Upper half-space]
Consider the upper half-space $\RR^{n-1} \times \RR_{\geq 0} $. The Dirichlet heat kernel is given by 
$$ K(x,y|t) = \frac{1}{4\pi t}\left(e^{-\frac{|x-y|^2}{4t}}- e^{-\frac{|x-\bar{y}|^2}{4t}} \right), $$
where $\bar{y}$ denotes the reflection of $y$ at the hyperplane $x_n = 0$. Consequently, integrating over $t$ we obtain the Green's function of the Helmholtz equation on the upper half-space, 
$$\int_0^\infty e^{-m^2t}K(x,y|t) = G^{\RR^{n-1} \times \RR_{\geq 0}}_{m^2}(x,y) = G^{\RR^n}_{m^2}(x,y) - G^{\RR^n}_{m^2}(x,\bar{y})$$
with $G^{\RR^n}$ given by \eqref{eq: Green's function Helmholtz R^n}. 
Notice that this example extends to any doublable Riemannian manifold $M$ with boundary $\gamma$ (doublable means that the metric on $M \cup_\gamma M$ is smooth) because on such a manifold, one can construct the heat kernel and the Helmholtz Green's function via image charges. In the language of \cite[Definition 1.1]{JR}, $M \cup_\gamma M$ is a manifold with a time reflection. This structure is relevant for reflection positivity.
\end{example}
\subsection{Dirichlet-to-Neumann operator} 
Let $M$ be a Riemannian manifold with boundary $\gamma$ and $\eta \in C^\infty_c(\gamma)$. For a function $f \in C^\infty(M)$, denote $\partial^n f \in C^\infty(\gamma)$ 
the outward normal derivative on the boundary of $f$. Denote by $\phi_\eta$ the solution to the Dirichlet problem for the Helmholtz operator, 
\begin{equation*}
\begin{cases}
    (\Delta + m^2)\phi_\eta = 0,  &  \\
    \phi_\eta|_{\partial M} = \eta &
\end{cases}
\end{equation*}
{
specified by 
\begin{equation}
    \phi_\eta = \int_{\gamma} \partial^n_yG_{m^2}(x,y)\eta(y) dy.
\end{equation}
}
The Dirichlet-to-Neumann operator is defined by 
\begin{equation*}
 \begin{array}{cccc}
    \DN^{\gamma,M}_{m^2} \colon & C_c^\infty(\gamma) & \to & C^\infty(\gamma)\\ &\eta &\mapsto &\partial^n\phi_\eta
 \end{array}
\end{equation*} 
This operator plays an important role in Riemannian geometry and its applications, and has been studied extensively in the literature in different directions, for instance in inverse problems after Calderon \cite{Calderon}, via its appearance in the Burghela-Friedlander-Kappeler gluing formula \cite{BFK}, or in scattering problems (as in \cite{HSW} and references therein). While there is plenty of literature treating the case of compact $\gamma$, surprisingly little works study this operator for non-compact boundary - see e.g. \cite{Lassas}, which contains a different (but most likely equivalent) definition from ours. In this paper, in order to prove our results we will assume that $\gamma$ is compact, although we present examples in which this assumption is not necessary. 
We now briefly summarize some of its known properties in the compact case. 

For $\gamma$ compact, one can show that it is an elliptic pseudo-differential operator of order 1. {
For  $\operatorname{Re} m^2 > 0$ (or more generally, for $-m^2$ not contained in the spectrum of $\Delta$), it is invertible.}\footnote{See \cite[Appendix C.1]{Taylor}. For $m =0$, $\DN^{\gamma,M}_0$ is instead an index 0 Fredholm operator with kernel given by constant functions and image given by functions of vanishing integral.} 
 For real $m$, $\DN^{\gamma,M}_{m^2}$ considered as an unbounded operator on $L^2(\gamma)$ is positive definite and therefore symmetric, and in fact self-adjoint \cite[Theorem 2]{Hoermander}.  For $m > 0$ and $\gamma$ compact, the spectrum of $\DN^{\gamma,M}_{m^2}$ is a sequence of positive real numbers 
$$0 < \lambda_1 \leq \lambda_2 \leq \ldots \leq \lambda_N \leq \ldots $$ 
with $\lambda_N \to \infty$ as $N \to \infty$. This spectrum is also known as the \emph{Steklov spectrum}, see for instance the historical survey \cite{Kuz14} or the review article \cite{GP}. 
Its inverse 
\begin{equation*}
    (\DN^{\gamma,M}_{m^2})^{-1} \colon L^2(\gamma) \to L^2 (\gamma)
\end{equation*}
is a bounded elliptic pseudodifferential operator of order -1, which is also positive and self-adjoint. Again, it will be important for us to consider the case of complex $m^2$ with $\operatorname{Re}m^2>0$. In this case, $\DN^{\gamma,M}_{m^2}$ remains an invertible pseudodifferential operator of order 1, which depends holomorphically on $m^2$ \cite{Carron}. 
\subsubsection{Integral kernel}
Let $M$ be a Riemannian manifold with boundary $\gamma$ and $G^{M,\gamma}_{m^2}$  the Green's function of $\Delta^M + m^2$ with Dirichlet boundary conditions. Then the harmonic extension $\phi_\eta$ of $\eta \in C^\infty(\gamma)$ is given by 
\begin{equation*}
    \phi_\eta(x) = \int_{y \in \gamma} \partial^n_y G_{m^2}(x,y)\, \eta(y) dy.
\end{equation*}
Therefore, the Dirichlet-to-Neumann operator can be written as 
\begin{equation}\label{eq: int kernel DN}
    (\DN^{\gamma,M}_{m^2}\eta)(x) = \int_{y \in \gamma} \partial^n_x\partial^n_yG^{M,\gamma}_{m^2}(x,y)\,\eta(y) dy.
\end{equation}
Thus, the integral kernel of the Dirichlet-to-Neumann operator is 
\begin{equation*} \DN^{\gamma,M}_{m^2}(x,y) = \partial^n_x\partial^n_yG^{M,\gamma}_{m^2}(x,y)
\end{equation*}
understood in a regularized sense as in \cite[Remark 3.4]{KMW}.
\subsubsection{Total Dirichlet-to-Neumann operator}
 Given a Riemannian manifold with a decomposition $M = M_1 \cup_\gamma M_2$, we define the \emph{total} or \emph{interface} Dirichlet-to-Neumann operator 
\begin{equation*}
    \DD^{\gamma,M}_{m^2}  = \DN^{\gamma,M_1}_{m^2} + \DN^{\gamma,M_2}_{m^2}.
\end{equation*}
Then we have the following elementary result on the integral kernel of its inverse: 
\begin{proposition}\label{prop: kernel DN inverse}
    If $\gamma$ is compact, the integral kernel of the inverse of the total Dirichlet-to-Neumann operator coincides with the restriction of the Green's function on $M$ to $\gamma$, that is, for $\eta \in C^\infty(\gamma)$ we have 
    \begin{equation*}
         ((\DD^{\gamma,M}_{m^2})^{-1}\eta)(x) = \int_{y\in \gamma}G^M_{m^2}(x,y)\eta(y)dy.
    \end{equation*}
\end{proposition}
This is a well-known fact (see for instance \cite[Proof of Theorem 2.1]{Carron}) but for completeness we give a proof in Appendix \ref{app: Proof of Prop}. 
This fact is also used
in the context of 
the Gaussian Free Field in constructive field theory, see e.g. \cite{Pickrell}, \cite{Lin} and \cite{GKRV}.  
\section{A conjecture
on the bulk dependence of the Dirichlet-to-Neumann operator}\label{sec: conj}
For a single Riemannian manifold $M$ with boundary $\gamma$, split the  Dirichlet-to-Neumann operator as 
\begin{equation*}
\DN^{\gamma,M}_{m^2} = \sqrt{\Delta^\gamma + m^2} + (\DN^{\gamma,M}_{m^2})'.
\end{equation*}
 If $\gamma$ is compact,  $(\DN^{\gamma,M}_{m^2})' $ is a pseudo-differential operator of order at most 0, see \cite{Taylor} or \cite[Lemma 2.4]{K}.  It follows that $(\DN^{\gamma,M}_{m^2})'$ is bounded by the Calderon-Vaillancourt theorem \cite{CV}. 
 In dimension 2, one actually has that $(\DN^{\gamma,M}_{m^2})'$ is a pseudo-differential operator of order at most -2 \cite[Proposition A.3]{KMW}. In the case $m=0$, conformal invariance implies that the remainder $(\DN^{\gamma,M}_{m^2})'$ is actually a smoothing operator. See \cite{Langlands}, \cite{GKRV} or \cite[Remark A.5]{KMW}.

\begin{remark}\label{rem:SIcyl} The term $\sqrt{\Delta^\gamma+m^2}$ is the DN operator for the semi-infinite cylinder $\gamma\times [0,+\infty)$. To see this, consider an eigenfunction 
$\eta$ of $\Delta^\gamma+m^2$ with eigenvalue $\lambda^2, \lambda >0$. Then $e^{-\lambda t}\eta$ is harmonic on $\gamma \times [0,\infty)$ and $$\DN_{m^2}^{\gamma,\gamma\times [0,\infty)}\eta = -\partial_t(e^{-t\lambda})\eta\big|_{t=0} = \lambda\eta = \sqrt{\Delta^\gamma+m^2}\;\eta.$$ 
It is easy to see all harmonic functions on the cylinder are of this kind and hence $\DN^{\gamma,\gamma \times [0,\infty)}_{m^2} = \sqrt{\Delta^\gamma+m^2}$. 
\end{remark} 
This means that to leading order, the Dirichlet-to-Neumann operator is only sensitive to the boundary geometry (it is given by the square root of the Helmholtz operator) but subleading orders depend on the bulk geometry.\footnote{In fact changing the metric away from the boundary changes the Dirichlet-to-Neumann operator by a smoothing operator.} The precise nature of this dependence has been the subject of much recent research. To prove our second gluing formula however we require a strong control over this dependence that we were unable to locate in the literature, namely the following statement: 
\begin{assumption}\label{conj: bound}
Let $M$  be a Riemannian manifold with boundary $\gamma$. 
We assume that:
\begin{enumerate}[(i)]
\item The operator $(\DN^{\gamma,M}_{m^2})'\colon L^2(\gamma) \to L^2(\gamma)$ is bounded.\footnote{If $\gamma$ is compact, this property is automatic, see above.}
\item 
For any $ \delta > 0 $ there is $C \gg 1$ such that  for $\operatorname{Re} m^2  > C$ we have 
\begin{equation}
    ||(\DN^{\gamma,M}_{m^2})'|| <\delta \cdot |m|. \label{eq: DN prime bound}
\end{equation}
\end{enumerate}
\end{assumption} 

I.e. for very large $m^2$ the operators $\sqrt{\Delta^\gamma + m^2}$ and $\DN^{\gamma,M}_{m^2}$ can be made arbitrarily close relative to $m$. 
If $M = M_1 \cup_\gamma M_2$ the interface Dirichlet-to-Neumann operator 
\begin{equation}\label{DN total}
    \DD^{\gamma,M}_{m^2}  = \DN^{\gamma,M_1}_{m^2} + \DN^{\gamma,M_2}_{m^2}
\end{equation}
 has the splitting 
\begin{equation*}
    \DD^{\gamma,M}_{m^2} = \underbrace{2\sqrt{\Delta^\gamma + m^2}}_{A} + \underbrace{(\DN^{\gamma,M_1}_{m^2})' + (\DN^{\gamma,M_2}_{m^2})'}_{\DD'} = A + \DD'. 
\end{equation*}
Assumption \ref{conj: bound} then implies that $\DD'$ is bounded and $||\DD'|| < 2|m|$ for very large $m^2$. 


One class of manifolds $M$ for which the Assumption holds is cylinders. For more explicit examples, see Appendix \ref{appendix: DN examples for Conj}.
\begin{example}[Cylinder]\label{example: cylinder}
    Let $\gamma$ be a closed 
    Riemannian manifold and $M = [0,L] \times \gamma$ with the product metric, and let $0 = \omega_0 \leq \omega_1 \leq \omega_2 \leq $ 
    denote the 
    eigenvalues of the Laplace operator $\Delta^\gamma$, with $(\eta_k)_{k=0}^\infty$ the corresponding eigenbasis of $L^2(\gamma)$. Let us for simplicity consider the $0-0$ block of the Dirichlet-to-Neumann operator. The unique function $\phi_k(t,x)$  
    on $I \times \gamma$ satisfying 
    $$(\Delta^{I\times \gamma}+m^2)\phi_k(t,x) = \left(-\frac{\partial^2}{\partial t^2 } + \Delta^\gamma + m^2\right)\phi_k(t,x) = 0$$ with $\phi_k(0,x) = \eta_k(x)$ and $\phi_k(L,x) \equiv 0$ is 
    $$ \phi_k(t,x) = \frac{\sinh (L-t)\mu_k}{\sinh L\mu_k}\eta_k(x),$$
    where $\mu_k = \sqrt{m^2 + \omega_k}$. Taking derivative with respect to $t$ at 0 we get 
    \begin{equation*}
        \DN^{\gamma,I\times\gamma}_{m^2}\eta_k = (\mu_k\coth L\mu_k) \eta_k.
    \end{equation*}
    The eigenvalues of the Dirichlet-to-Neumann operator are therefore given by 
    $\lambda_k = \mu_k\coth(L\mu_k) = \sqrt{m^2 + \omega_k}\coth(L\sqrt{m^2 + \omega_k})$. 
    In particular, both the Dirichlet-to-Neumann operator and $\sqrt{\Delta^\gamma+m^2}$ are diagonal in the basis $\eta_k$ and the difference of the eigenvalues is 
    \begin{equation*}
        \lambda_k - \mu_k = \mu_k(\coth L\mu_k - 1) \underset{\mu_k \to \infty}{\simeq} 2\mu_k e^{-2L\mu_k} .
    \end{equation*}
    As $k\to \infty$ the eigenvalues of this operator decay exponentially, since $\mu_k \to \infty$ as $k \to \infty$. Therefore this operator is bounded. As $m \to \infty$ all eigenvalues are exponentially suppressed and Assumption \ref{conj: bound} holds. 
\end{example}
In fact, Assumption \ref{conj: bound} holds whenever the manifold is cylindrical close to the boundary. 
\begin{lemma}\label{lemma: Assump for product metric}
    If $M$ is a Riemannian manifold with compact boundary $\gamma$ such that the metric is the product metric near $\gamma$ (i.e. there is a neighborhood of the boundary which isometric to $\gamma\times [0,\epsilon]$), then Assumption \ref{conj: bound} holds.
\end{lemma}
We give the proof in Appendix \ref{appendix: proof of Lemma}. Compactness of $\gamma$ is used in the estimate (\ref{A1 computation 5}).

\begin{remark}\label{rem: Conj when DN and A commute}
   If $\DN^{\gamma,M}_{m^2}$ and $\sqrt{\Delta^\gamma + m^2}$ commute (for all $m >0$), the eigenvalues estimates derived in \cite{GKLP} (after H\"ormander \cite{Hoermander}) imply a stronger version of Assumption \ref{conj: bound}, but only in the case where $m^2 > 0$ (i.e. $\operatorname{Im} m^2 = 0$). Namely, let $\lambda_{k,m^2}$ be the eigenvalues of $\DN^{\gamma, M}_{m^2}$ and $\omega_k$ the eigenvalues of $\Delta^\gamma$. Then we have \cite[Theorem 4.2]{GKLP} 
   \begin{equation*}
       |\lambda_{k,m^2} - \sqrt{\omega_k + m^2}| < C,
   \end{equation*}
   where $C$ is a constant independent of $k$ and $m$.  
   In particular, for $m > C/\delta$, the bound \eqref{eq: DN prime bound} holds. 
   
   On the other hand, the arguments in \cite{GKLP} do not immediately extend to complex $m^2$, and
it is also proven in \cite{GKLP} that, for $m = 0$ and $M$ is a domain in $\RR^n$ with connected boundary,  $n \geq 3$,  the operators $\DN^{\gamma,M}$ and $\sqrt{\Delta^\gamma}$ commute if and only if $M$ is a ball. This result suggests that $\DN_{m^2}^{\gamma,M}$ and $\sqrt{\Delta^\gamma+m^2}$ in general do not commute. 

In the general case, a version of Assumption \ref{conj: bound} was proven in \cite{Vodev}. However, this result requires the mass to have a large imaginary part.
\end{remark}

\begin{conjecture}\label{conjecture}
    Assumption \ref{conj: bound} holds for any 
    Riemannian manifold $M$ with (not necessarily compact) 
    boundary.
\end{conjecture}

\section{Gluing formulae}\label{sec: gluing formula}
\subsection{Gluing formula I}
Consider a  
Riemannian manifold $M$, with (possibly empty) boundary $\dd M$, with $\gamma\subset M$ a 
{closed}
codimension $1$ submanifold disjoint from $\dd M$, such that one has  a decomposition 
\begin{equation}\label{M splitting sec 4}
M = M_1 \cup_\gamma M_2.
\end{equation}
 The Green's function and heat kernel on $M_i$ are denoted by $G^{(i)}$ and $K^{(i)}$ respectively, with Dirichlet boundary condition imposed on all boundary components, operators on $M$ carry no superscript. The gluing formula for the Green's function is: 
\begin{proposition}
    Let $x\in M_a$, $y\in M_b$, $a,b\in \{1,2\}$. Then:
\begin{equation}
    G_{m^2}(x,y) = \delta_{ab}G^{(a)}_{m^2}(x,y) + \int_{\gamma \times \gamma} \partial^n_uG^{(a)}_{m^2}(x,u)\DD^{-1}_{m^2}(u,v)\partial^n_vG^{(b)}_{m^2}(v,y) du dv .\label{eq: gluing green's functions}
\end{equation}
The first term in the r.h.s. is absent for $a\neq b$.
\end{proposition}

A proof was given \cite[Section 4.2]{KMW}, where we refer to for further discussion. Note that this formula holds for all  $m^2 \notin  -\sigma(\Delta^X)$, in particular for $\operatorname{Re} m^2 >0$. 

Our first gluing formula for the heat kernel is the following.  
\begin{proposition}[Gluing formula for heat kernels I]\label{prop: gluing heat kernel 1}
One has the following gluing formula for heat kernels. Let $x\in M_a$, $y\in M_b$, $a,b\in \{1,2\}$. Then:
\begin{multline}\label{eq: gluing heat kernel 1}
    K(x,y|t) = \delta_{ab}K^{(a)}(x,y|t) +  \\
    \int_{t_0+t_1+t_2 = t, t_i > 0}dt_0dt_1\, \int_{\gamma \times \gamma} dudv\, \partial^n_uK^{(a)}(x,u|t_0)\LL^{-1}[\DD^{-1}_{m^2}(u,v)](t_1)\partial^n_vK^{(b)}(v,y|t_2).
\end{multline}
\end{proposition}
\begin{figure}[h]
    \centering
   \begin{tikzpicture}
   \node[bulk] (b1) at (-2,-1) {}; 
   \node[bulk] (b2) at (-1.5,1) {};
   \node[coordinate, label={$x$}] at (b1.west) {};
    \node[coordinate, label={$y$}] at (b2.west) {};
 \node[coordinate, label={$\gamma$}] at (0,-2.5) {};
  \node[coordinate, label={$M_1$}] at (-.5,1.7) {};
  \node[coordinate, label={$M_2$}] at (.5,1.7) {};
   \draw[bubu] (b1) -- (b2) node[midway,right] {$t$};
    \draw (0,2) -- (0,-2);
    \node at (1,0) {=};
    \begin{scope}[xshift = 5cm]
    \node[bulk] (b1a) at (-2,-1) {};
   \node[bulk] (b2a) at (-1.5,1) {};
   \node[coordinate, label={$x$}] at (b1a.west) {};
    \node[coordinate, label={$y$}] at (b2a.west) {};
      \node[coordinate, label={$M_1$}] at (-.5,1.7) {};
  \node[coordinate, label={$M_2$}] at (.5,1.7) {};
   \draw[bubu1] (b1a) -- (b2a) node[midway,right] {$t$};
    \draw (0,2) -- (0,-2);
     \node[coordinate, label={$\gamma$}] at (0,-2.5) {};

    \node at (1,0) {+};
     \begin{scope}[xshift = 4cm]
     \node[bulk] (b1b) at (-2,-1) {};
   \node[bulk] (b2b) at (-1.5,1) {};
   \node[coordinate, label={$x$}] at (b1b.west) {};
    \node[coordinate, label={$y$}] at (b2b.west) {};
      \node[coordinate, label={$M_1$}] at (-.5,1.7) {};
  \node[coordinate, label={$M_2$}] at (.5,1.7) {};
   \node[bdry] (b3b) at (0,-.75) {};
   \node[bdry] (b4b) at (0,.75) {}; 
    \node[coordinate, label={$u$}] at (b3b.east) {};
    \node[coordinate, label={$v$}] at (b4b.east) {};
           \draw[bubo] (b1b) -- (b3b) node[midway,above] {$t_0$};
           \draw[bobo]( b3b) -- (b4b)node[midway,right] {$t_1$}; 
           \draw[bubo](b4b) -- (b2b) node[midway,above] {$t_2$};
    \draw (0,2) -- (0,-2);
     \node[coordinate, label={$\gamma$}] at (0,-2.5) {};
    \end{scope}
    \end{scope}
\end{tikzpicture}
    \caption{Pictorial description of the gluing formula for the heat kernel. The fat straight line to the left denotes $K(x,y|t)$, the thin straight line in the middle $K^{(1)}(x,y|t)$, a red squiggly edge denotes $\partial^n_u K^{(1)}(x,u|t)$, a curly edge denotes $\LL^{-1}[\DD_{m^2}^{-1}(u,v)](t)$. Grey vertices are integrated over the interface $\gamma$.}
    \label{fig: gluing formula 1}
\end{figure}
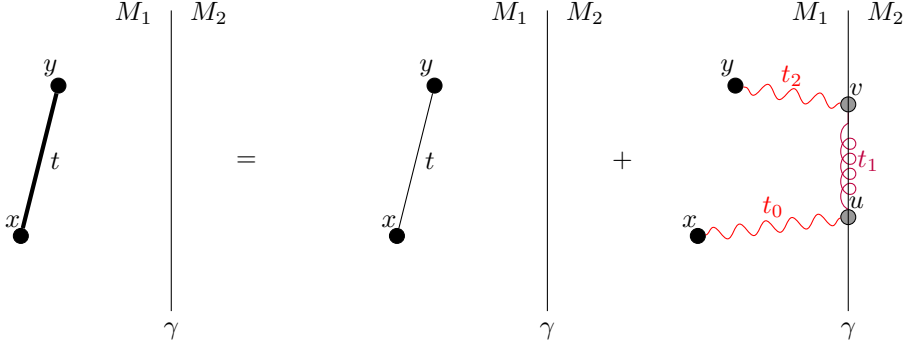

\begin{proof}
    By the arguments above we have $\LL[K(x,y|t)](m^2) = G_{m^2}(x,y)$, and similarly for $K^{(i)},G^{(i)}$. Since the convergence is uniform in $(x,y)$ on compact subsets of $x \neq y$ we can exchange the Laplace transform with the normal derivative to get 
    $$\LL[\partial^n_u K(x,u|t)](m^2) = \partial^nG_{m^2}(x,u).$$
    Finally, we note again that for $u,v \in \gamma$, by Proposition \ref{prop: kernel DN inverse} one has $\DD^{-1}_{m^2}(u,v) = G_{m^2}(u,v)$. Therefore, 
    $$\LL[K(u,v|t)](m^2) = \DD^{-1}_{m^2}$$
    in particular $\LL^{-1}[\DD^{-1}_{m^2}]$ is defined and the gluing formula follows by applying $\LL^{-1}$ to \eqref{eq: gluing green's functions} and using again the fact that multiplication is sent to convolution under the inverse Laplace transform. 
\end{proof}

\begin{remark} \label{rem: path integral}
     The proof of Proposition \ref{prop: gluing heat kernel 1} shows that an alternative (slightly circular) way of writing of the gluing formulae \eqref{eq: gluing green's functions} and \eqref{eq: gluing heat kernel 1} is 
    \begin{equation}
        G_{m^2}(x,y) = \delta_{ab}G^{(a)}_{m^2}(x,y) + \int_{\gamma \times \gamma} \partial^n_uG^{(a)}_{m^2}(x,u)G_{m^2}(u,v)\partial^n_vG^{(b)}_{m^2}(v,y) du dv  \label{eq: gluing greens functions 2}
    \end{equation} 
    and \begin{multline}
         K(x,y|t) = \delta_{ab}K^{(a)}(x,y|t) +  \\                                          
    \int_{t_0+t_1+t_2 = t, t_i > 0}dt_0dt_1\, \int_{\gamma \times \gamma} dudv\, \partial^n_uK^{(a)}(x,u|t_0)K(u,v|t_1)\partial^n_vK^{(b)}(v,y|t_2).\label{eq: gluing heat kernels 2}
    \end{multline} 
    These formulae admit natural interpretations in terms of path integrals and Wiener measure, see Appendix \ref{app: path integral}. 
\end{remark}

\subsection{An expression for $\LL^{-1}[\DN^{-1}]$. } 
 As we will see presently, it is possible to rewrite the formula (\ref{eq: gluing heat kernels 2}) to contain only the heat kernels of $M_i$ and $\gamma$.
Suppose that $A$ and $B$ are two operators on a Hilbert space such that $A+B$ and $A$ are invertible and such $A^{-1}B$ is bounded of norm $||A^{-1}B|| < 1$ (note that $A,B$ are not required to be bounded). Then we have the following geometric progression formula for $(A+B)^{-1}$:
\begin{multline}\label{eq: geometric progression}
    (A+B)^{-1} = (A(I +A^{-1}B)) ^{-1} = (I+A^{-1}B)^{-1}A^{-1} = \\
    =\sum_{k\geq 0}(-A^{-1}B)^kA^{-1} = A^{-1} - A^{-1}BA^{-1} + \ldots 
\end{multline} 
where the  series converges in operator norm. \\
We want to apply this result to compute the inverse of the total Dirichlet-to-Neumann operator of the glued manifold $M = M_1 \cup_\gamma M_2$ via the splitting $\DD^{\gamma,M}_{m^2} = A + \DD'$ (in the notations of Section \ref{sec: conj} above). By Assumption \ref{conj: bound}, we have  that for large $m$, $||\DD'|| < 2\delta |m|$. On the other hand, $A^{-1}$ is a bounded operator with norm $||A^{-1}|| = \frac{1}{2|m|}$ and therefore we have $||A^{-1}\DD'|| < \delta < 1.$ By \eqref{eq: geometric progression}, we then have 
\begin{equation}\label{eq: geom progression for DN^-1}
    \DD^{-1}_{m^2} = \sum_{k \geq 0} (-A^{-1}\DD')^{k}A^{-1} = A^{-1} - A^{-1}\DD'A^{-1} + \ldots
\end{equation} 
We can compute the inverse Laplace transform of $\DD^{-1}_{m^2}$ by applying it to $A^{-1}$ and $\DD'$ separately, and then make use the fact that multiplication gets mapped to convolution. Namely, we have by Proposition \ref{prop: kernel DN inverse}
\begin{equation*}
    A^{-1} = (\DD^{\gamma,\gamma \times \RR})^{-1} 
    = G^{\gamma \times \RR}\rvert_{(\gamma \times \{0\}) \times (\gamma \times \{0\})}
\end{equation*}
and therefore 
\begin{multline*}
    \LL^{-1}[A^{-1}(u,v)](t) = K^{\gamma\times \RR}((u,0),(v,0)|t)  =\\
=K^\gamma(u,v|t)K^\RR(0,0|t) = \frac{1}{\sqrt{4\pi t}}K^\gamma(u,v|t).
\end{multline*}
For 
\begin{equation}\label{DD'}
\DD' = \DD^{\gamma,M}_{m^2} - A
\end{equation}
we then have
\begin{equation*}
    \LL^{-1}[\DD'
    (u,v)](t) = \partial^n_u\partial^n_v \Big(K(u,v|t) - K^{\gamma \times \RR_{\geq 0}}(u,v|t) \Big)
\end{equation*}
and therefore 
\begin{multline*}
    \LL^{-1}[\DD'
    (u,v)](t) = \\
    =\partial^n_u\partial^n_v \left(K^{(1)}(u,v|t) + K^{(2)}(u,v|t) - 2K^{\gamma \times \RR_{\geq 0}}((u,0),(v,0)|t)\right) = \colon \Kp(u,v|t).
\end{multline*}

Using again the fact that the heat kernel on a product manifold is the product of heat kernels, we can slightly simplify 
\begin{multline*}
    \partial^n_u\partial^n_v \Big(K^{(a)}(u,v|t) - K^{\gamma \times \RR_{\geq 0}}(u,v|t) \Big) =\\
    =\partial^n_u\partial^n_v K^{(a)}(u,v|t) - K^{\gamma }(u,v|t)\partial^n_x\partial^n_y K^{\RR_{\geq 0}}(x,y|t) 
    = \partial^n_u\partial^n_v K^{(a)}(u,v|t) - \frac{1}{\sqrt{4\pi t^3}} K^{\gamma }(u,v|t)
\end{multline*}
with $a\in\{1,2\}$,
and thus we get the following formula for $\Kp$: 
\begin{equation*}
    \Kp(u,v|t) =  \partial^n_u\partial^n_v \left(K^{(1)}(u,v|t) + K^{(2)}(u,v|t) \right)- \frac{2}{\sqrt{4\pi t^3}}K^{\gamma }(u,v|t).
\end{equation*}

\subsection{Gluing formula II}
The argument above proves our main result.

\begin{thm}[Gluing formula for heat kernels II]\label{thm: gluing formula}
Let $M$ be a Riemannian manifold with decomposition (\ref{M splitting sec 4}). If Assumption \ref{conj: bound} holds for $M_{1,2}$ (for example, if they have product metric near $\gamma$ and $\gamma$ is compact), one has the following gluing formula for heat kernels. Let $x\in M_a$, $y\in M_b$, $a,b\in \{1,2\}$. Then:
\begin{multline}\label{gluing formula}
    K(x,y|t) = \delta_{ab} K^{(a)}(x,y|t) + \\
    +\sum_{k\geq 0}(-1)^k\int_{\sum_{i=0}^{2k+2} t_i = t, t_i> 0}\prod_{i=0}^{2k+1}dt_i\int_{\gamma^{2k+2}}\prod_{i=0}^{2k+1} du_i\,
    \partial^n_{u_0}K^{(a)}(x,u_0|t_0) \cdot \\ 
    \cdot \prod_{i=0}^{k-1}K^{\gamma \times \RR}((u_{2i},0),(u_{2i+1},0)|t_{2i+1})\Kp(u_{2i+1},u_{2i+2}|t_{2i+2})\cdot \\
    \cdot K^{\gamma \times \RR}((u_{2k},0),(u_{2k+1},0)|t_{2k+1})\partial^n_{u_{2k+1}}K^{(b)}(u_{2k+1},y|t_{2k+2}).
\end{multline}
\end{thm}
Using again the factorization of heat kernels on direct products, the following is an immediate corollary: 
\begin{corollary}\label{cor: gluing formula 2}
Under the same assumptions as in Theorem \ref{thm: gluing formula}, we have:
\begin{multline}
\label{gluing formula 2}
    K(x,y|t) = \delta_{ab} K^{(a)}(x,y|t) + \\
    +\sum_{k\geq 0}(-1)^k\int_{\sum_{i=0}^{2k+2} t_i = t, t_i> 0}\prod_{i=0}^{2k+1}dt_i\int_{\gamma^{2k+2}}\prod_{i=0}^{2k+1} du_i\,
    \partial^n_{u_0}K^{(a)}(x,u_0|t_0) \cdot \\ 
    \cdot \prod_{i=0}^{k-1}\frac{1}{\sqrt{4\pi t_{2i+1}}}K^{\gamma}(u_{2i},u_{2i+1}|t_{2i+1})\Kp(u_{2i+1},u_{2i+2}|t_{2i+2})\cdot \\
    \frac{1}{\sqrt{4\pi t_{2k+1}}}\cdot K^{\gamma}(u_{2k},u_{2k+1}|t_{2k+1})\partial^n_{u_{2k+1}}K^{(b)}(u_{2k+1},y|t_{2k+2}).
\end{multline} 
\end{corollary}
We remark that in \eqref{gluing formula 2}, the heat kernel $K$ of the glued manifold is expressed solely in terms of the heat kernels on $\gamma$ and $M^{(a)}$.
\begin{figure}[h]
    \centering
   \begin{tikzpicture}
   \node[bulk] (b1) at (-2,-1) {}; 
   \node[bulk] (b2) at (-1.5,1) {};
   \node[coordinate, label={$x$}] at (b1.west) {};
    \node[coordinate, label={$y$}] at (b2.west) {};
 \node[coordinate, label={$\gamma$}] at (0,-2.5) {};
  \node[coordinate, label={$M_1$}] at (-.5,1.7) {};
  \node[coordinate, label={$M_2$}] at (.5,1.7) {};
   \draw[bubu] (b1) -- (b2) node[midway,right] {$t$};
    \draw (0,2) -- (0,-2);
    \node at (1,0) {=};
    \begin{scope}[xshift = 5cm]
    \node[bulk] (b1a) at (-2,-1) {};
   \node[bulk] (b2a) at (-1.5,1) {};
   \node[coordinate, label={$x$}] at (b1a.west) {};
    \node[coordinate, label={$y$}] at (b2a.west) {};
      \node[coordinate, label={$M_1$}] at (-.5,1.7) {};
  \node[coordinate, label={$M_2$}] at (.5,1.7) {};
   \draw[bubu1] (b1a) -- (b2a) node[midway,right] {$t$};
    \draw (0,2) -- (0,-2);
     \node[coordinate, label={$\gamma$}] at (0,-2.5) {};

    \node at (1,0) {$+\sum_{k\geq 0}$};
     \begin{scope}[xshift = 4cm]
     \node[bulk] (b1b) at (-2,-1) {};
   \node[bulk] (b2b) at (-1.5,1) {};
   \node[coordinate, label={$x$}] at (b1b.west) {};
    \node[coordinate, label={$y$}] at (b2b.west) {};
    \node[coordinate, label={$M_1$}] at (-.5,3.3) {};
  \node[coordinate, label={$M_2$}] at (.5,3.3) {};
   \node[bdry] (b3b) at (0,-3) {};
   \node[bdry] (b5b) at (0,-2) {};
   \node[bdry] (b6b) at (0,-1) {};
   \node[bdry] (b7b) at (0,0) {};
   \node[bdry] (b8b) at (0,1) {}; 
   \node[bdry] (b9b) at (0,2) {};
   \node[bdry] (b4b) at (0,3) {}; 
    \node[coordinate, label=right:{$u_0$}] at (b3b.east) {};
        \node[coordinate, label=right:{$u_1$}] at (b5b.east) {};
    \node[coordinate, label=right:{$u_2$}] at (b6b.east) {};
    \node[coordinate, label=right:{$u_3$}] at (b7b.east) {};
    \node[coordinate, label=right:{$u_{i}$}] at (b8b.east) {};
    \node[coordinate, label=right:{$u_{2k}$}] at (b9b.east) {};
    \node[coordinate, label=right:{$u_{2k+1}$}] at (b4b.east) {};
           \draw[bubo] (b1b) -- (b3b) node[midway,above] {$t_0$};
           \draw[bobo2](b3b) -- (b5b)node[midway,right] {$t_1$}; 
            \draw[bobo3](b5b) -- (b6b)node[midway,right] {$t_2$}; 
            \draw[bobo2](b6b) -- (b7b)node[midway,right] {$t_3$};
\draw[dotted] (b7b) to[dashed, bend left=60] (b8b) {};
\draw[dotted] (b8b) to[dashed, bend right=60] (b9b) {};
\draw[bobo2] (b9b) -- (b4b)node[midway,right] {$t_{2k+1}$};
           \draw[bubo](b4b) -- (b2b) node[midway,above] {$t_{2k+2}$};
    \draw (0,3.3) -- (0,-3.3);
     \node[coordinate, label={$\gamma$}] at (0,-3.7) {};
    \end{scope}
    \end{scope}
\end{tikzpicture}
    \caption{Pictorial description of the gluing formula for the heat kernel. The fat straight line to the left denotes $K(x,y|t)$, the thin straight line in the middle $K^{(1)}(x,y|t)$, a red squiggly edge denotes $\partial^n_u K^{(1)}(x,u|t)$. Blue wavy edges denote $\frac{1}{\sqrt{4 \pi t}}K^\gamma(u,v|t)$, green zigzag edges denote $K'(u,v|t)$.  Grey vertices are integrated over the interface $\gamma$.}
    \label{fig: gluing formula 2}
\end{figure}
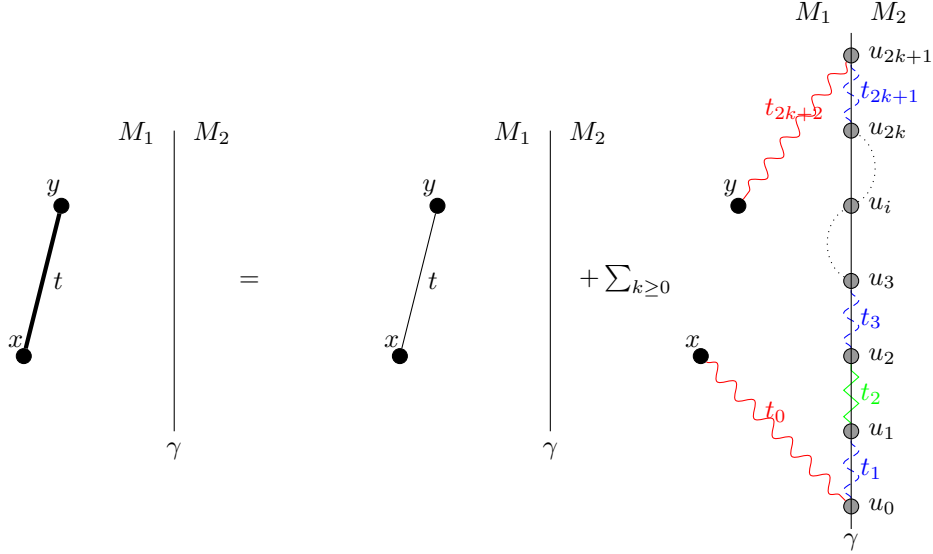
\subsection{``Cutting-out'' formulae for the heat kernel}\label{ss: cutting formulae}
    Results of a similar nature to equation \eqref{eq: gluing green's functions} are known in the literature on scattering theory as Birman-Krein formulae, see for the recent paper \cite{HSW} and references therein. These results are ``cutting formulae,'' i.e., they allow one to express the Green's function with Dirichlet boundary conditions on a closed region 
    $M_1 
    \subset M$ in terms of the full Green's function on $M$.\footnote{ 
    A typical example is: $M=\RR^n$ and $M_1=\overline{\RR^n\setminus M_2}$, with $M_2\subset \RR^n$ a compact subset (``defect'') with smooth $(n-1)$-dimensional  boundary $\gamma$. This setup is pertinent in scattering theory -- scattering on the defect $M_2$.}
    For instance, formula (19) in \cite{HSW} reads (in our notation) 
    \begin{equation}
        G_{m^2}(x,y) - \delta_{ab}G^{(a)}_{m^2}(x,y) = \int_{\gamma \times\gamma}G_{m^2}(x,u)\DD_{m^2}(u,v)G_{m^2}(v,y). \label{eq: cutting}
    \end{equation}
    By Proposition \ref{prop: kernel DN inverse}, $\DD_{m^2}$ can be expressed as the inverse of the operator given by restricting $G_{m^2}$ to $\gamma \times \gamma$, so the right hand side of equation \eqref{eq: cutting} can be expressed in terms of the full Green's function only. Applying the inverse Laplace transform, we obtain a ``cutting formula'' for the heat kernel: 
    \begin{multline*}
        K(x,y|t) - \delta_{ab}K^{(a)}(x,y|t)  = \\ 
        =\int_{t_0+t_1+t_2 = t, t_i > 0}dt_0dt_1\, \int_{\gamma \times \gamma} dudv\, K(x,u|t_0)\LL^{-1}[\DD_{m^2}(u,v)](t_1)K(v,y|t_2).
    \end{multline*}
    Here the term $\LL^{-1}[\DD_{m^2}]$ could be expanded similarly to the proof of Theorem \ref{thm: gluing formula}:\footnote{
    This sequence is convergent for $\re m^2>C$ under Assumption \ref{conj: bound} for $M_{1,2}$: one has 
    ${||A(\DD^{-1}-A^{-1})||}=||-\DD'(A+\DD')^{-1}||\leq ||\DD'|| \cdot ||(A+\DD')^{-1}||< 2\delta |m|\frac{1}{2|m|-2\delta |m|}=\frac{\delta}{1-\delta}$. Here $\DD'$ is as in (\ref{DD'}). So, choosing e.g. $\delta=\frac12$ we obtain $||A(\DD^{-1}-A^{-1})||<1$, hence the geometric series  (\ref{DD geometric series}) is convergent.
    }
    \begin{equation}\label{DD geometric series}
        \DD=(A^{-1}+(\DD^{-1}-A^{-1}))^{-1}=A-A(\DD^{-1}-A^{-1})A+A(\DD^{-1}-A^{-1})A(\DD^{-1}-A^{-1})A-\cdots
    \end{equation}
    Hence, one has 
    \begin{equation*}
        \LL^{-1}[\DD^{-1}]=\sum_{k\geq 0} (\LL^{-1}[A]*\LL^{-1}[\DD^{-1}-A^{-1}])^{*k} *\LL^{-1}[A],
    \end{equation*}
    where $*$ stands for the convolution in the $t$ variable. Passing to the integral kernels of the operators, this implies the second cutting formula
    \begin{multline}
         K(x,y|t) - \delta_{ab}K^{(a)}(x,y|t)  =\\
         =\sum_{k\geq 0}(-1)^k\int_{\sum_{i=0}^{2k+2}t_i=t,t_i>0} 
         \prod_{i=0}^{2k+1} dt_i \int_{\gamma^{2k+2}}\prod_{i=0}^{2k+1} du_i \cdot \\ \cdot K(x,u_0|t_0)\prod_{i=0}^{k-1}\frac{1}{\sqrt{\pi} t_{2i+1}^{3/2}}K^\gamma(u_{2i},u_{2i+1}|t_{2i+1}) \til{K}(u_{2i+1},u_{2i+2}|t_{2i+2})\cdot \\
         \cdot \frac{1}{\sqrt{\pi}t_{2k+1}^{3/2}} K^\gamma(u_{2k},u_{2k+1}|t_{2k+1}) K(u_{2k+1},y|t_{2k+2}),\label{eq: cutting II}
    \end{multline}
    where we denoted
    \begin{equation*}
        \til{K}(u,v|t) \colon = K(u,v|t)-\frac{1}{\sqrt{4\pi t}} K^\gamma(u,v|t).
    \end{equation*}
    Formula (\ref{eq: cutting II}) is similar in structure to the gluing formula (\ref{gluing formula 2}).
    

Note that equation \eqref{eq: cutting II} expresses the heat kernel with Dirichlet boundary conditions  in terms of the full heat kernel (and heat kernel on $\gamma$) only.

\section{Examples}\label{sec: examples}
\subsection{Gluing two rays}
The heat kernel on $\RR$ is 
\begin{equation*}
K(x,y|t) = \frac{1}{\sqrt{4\pi t}}e^{-\frac{(x-y)^2}{4t}}    
\end{equation*}
and the heat kernel on $\RR_{\geq 0}$ with Dirichlet boundary conditions is 
\begin{equation*}
    K^{(1)}(x,y|t) = \frac{1}{\sqrt{4\pi t}}\left(e^{-\frac{(x-y)^2}{4t}} - e^{-\frac{(x+y)^2}{4t}}\right).
\end{equation*}
The normal derivative 
of the heat kernel is 
\begin{equation*}
    - \partial_yK^{(1)}(x,y|t)\rvert_{y=0} = \frac{ -1}{2\sqrt{\pi}t^{3/2}}x e^{-\frac{x^2}{4t}}.
\end{equation*}
The Dirichlet-to-Neumann operator is $2m = 2 \sqrt{\Delta^\gamma + m^2} = A$, with $A^{-1} = (2m)^{-1}$ and its inverse Laplace transform is 
\begin{equation*}
    \LL^{-1}[(2m)^{-1}]  = \frac12\LL^{-1}[s^{-1/2}](t) =  \frac{1}{2\Gamma(1/2)t^{1/2}} = \frac{1}{\sqrt{4\pi t}}=K(0,0|t).
\end{equation*}
Equation \eqref{eq: gluing heat kernel 1} then becomes for $x,y \geq 0$
\begin{equation*}
\begin{aligned}
K(x,y|t) - &K^{(1)}(x,y|t) = \frac{1}{\sqrt{4\pi t}}e^{\frac{(x+y)^2}{4t}} \\
&= \int_{t_0+t_1+t_2=t, t_i > 0}dt_0dt_1 \underbrace{\frac{ -1}{\sqrt{4\pi t^3_0}}xe^{-\frac{x^2}{4t_0}}}_{\partial^n_uK_1(x,u|t_0)} \underbrace{\frac{1}{\sqrt{4\pi t_1}}}_{\LL^{-1}[\DD^{-1}]} \underbrace{\frac{ -1}{\sqrt{4\pi t_2^3}} ye^{-\frac{y^2}{t_2}}}_{\partial^n_vK_1(v,y|t_0)} \\
&= \frac{1}{8\pi^{3/2}}\int_{t_0+t_1+t_2=t, t_i > 0}dt_0dt_1\frac{xy\, e^{-\frac{x^2}{4t_0}-\frac{y^2}{4t_2}}}{\sqrt{t_0^3t_1t_2^3}}.
\end{aligned} 
\end{equation*}
\subsection{Gluing two finite intervals}
The heat kernel on a finite interval of length $L$ is 
\begin{equation}
\begin{aligned} \label{K interval}
    K^{L}(x,y|t)&= \frac{1}{\sqrt{4\pi t}}
    \sum_{k\in \ZZ} \left( e^{-\frac{1}{4t} (x-y+2k L)^2}- e^{-\frac{1}{4t} (x+y+2k L)^2} \right)
    \\ 
    &= \frac{2}{L}\sum_{k = 1}^\infty e^{-\frac{\pi^2k^2}{L^2} t}\sin\left(\frac{\pi k x}{L}\right)\sin\left(\frac{\pi k y}{L}\right).
\end{aligned} 
\end{equation}
Here the first equation is proved by the image charge method, and the second one by considering eigenfunctions of the Laplacian (with Dirichlet boundary conditions) on the interval $[0,L]$. Its normal derivative at 0 
is
\begin{equation}\label{interval normal derivative}
\begin{aligned}
    - \partial_y \rvert_{y=0} K^L(x,y|t) &= 
    \frac{ -1}{\sqrt{4\pi}\, t^{3/2}}\sum_{k\in \ZZ} (x+2kL) e^{-\frac{1}{4t}(x+2kL)^2}
    \\
    &=  -\frac{2}{L}\sum_{k=1}^\infty  e^{-\frac{\pi^2k^2}{L^2}  t}\frac{\pi k}{L}\sin\left(\frac{\pi k x}{L}\right).
\end{aligned} 
\end{equation}
The Dirichlet-to-Neumann operator is $\DD = m(\coth mL_1 + m\coth mL_2)$ we can compute the inverse Laplace transform of its inverse via residues, 
\begin{equation*}
\begin{aligned}
    \LL^{-1} [\DD^{-1}] &= \frac{1}{2\pi \ii}\int_C \frac{e^{st}ds}{\sqrt{s}(\coth{\sqrt{s}L_1}+ \coth \sqrt{s}L_2)} \\
    &= \frac{1}{2\pi i} \int_C \frac{e^{st } ds\,\sinh \sqrt{s}L_1\sinh \sqrt{s}L_2}{\sqrt{s}\sinh \sqrt{s}(L_1 + L_2)} .
\end{aligned}
\end{equation*}
Here the contour $C$ can be taken to be the imaginary axis, and we can close it at infinity,  adding an infinitely remote arc with $\mr{arg}(s)\in [\frac{\pi}{2},\frac{3\pi}{2}]$. Notice that the integrand is an even function of $\sqrt{s}$ and therefore can be extended over the branch cut of the square root. 
 This function has poles at $s = -\frac{(\pi k)^2}{(L_1 + L_2)^2}$ for $k =1,2,3,\ldots .$ 
 So, computing the  residues, the inverse Laplace transform is  
 \begin{equation*}
 \begin{aligned}
     \LL^{-1} [\DD^{-1}](t) = 2\sum_{k\geq 1} \frac{(-1)^{k+1}}{L_1 + L_2} e^{-\frac{\pi^2k^2 t}{(L_1+L_2)^2}}\sin\frac{\pi k L_1}{L_1 + L_2}\sin\frac{\pi k L_2}{L_1 + L_2}\\
     = \frac{1}{\sqrt{4\pi t}} \sum_{n\in\ZZ} \left(e^{-\frac{(L_1+L_2)^2 n^2}{t}} - 
     e^{-\frac{((n+1)L_1+n L_2)^2}{t}}\right).
 \end{aligned}
 \end{equation*}
 The second line here is obtained by Poisson resummation of the first.
 The gluing formula (\ref{eq: gluing heat kernel 1}) then becomes (for $0\leq x,y \leq L_2$)
 \begin{multline*}
     K^L(L_1+x,L_1+y|t) - K^{L_2}(x,y|t) = \frac{1}{4\pi}\int_{t_0+t_1+t_2 = t, t_i> 0}dt_0dt_1\frac{1}{(t_0t_2)^{3/2}} \\
     \sum_{k_0\in \ZZ} \left(x+ 2 k_0 L_2\right)e^{-\frac{1}{4t_0}\left(x + 2k_0 L_2\right)^2}
     \cdot 2 \sum_{k_1\geq 1} \frac{(-1)^{k_1+1}}{L_1 + L_2} e^{-\frac{\pi^2k_1^2 t_1}{(L_1+L_2)^2}}\sin\frac{\pi k_1 L_1}{L_1 + L_2}\sin\frac{\pi k_1 L_2}{L_1 + L_2}\\ 
     \cdot \sum_{k_2\in \ZZ} \left(y+ 2k_2 L_2\right)e^{-\frac{1}{4t_2}\left(y + 2k_2 L_2\right)^2}.
 \end{multline*}
 
Alternatively, we can apply the decomposition \begin{equation}\DD = \underbrace{2m}_A + \underbrace{m(\coth (mL_1) -1 + \coth (mL_2) -1 )}_{\DD'}.\end{equation}
One has 
\begin{multline*}
 \LL^{-1}[m (\coth (mL)-1)]=\dd^n_x|_{x=0}\dd^n_y|_{y=0}(K^L(x,y|t)-K^\infty(x,y|t))\\
 =2\sum_{k\geq 1} \frac{1}{\sqrt{4\pi} t^{3/2}} \left(1-\frac{2k^2L^2}{t}\right) e^{-\frac{k^2L^2}{t}}  
\end{multline*}
and
\begin{equation} \label{eq35}
     \LL^{-1}\left[\frac{1}{2m}\right]*\LL^{-1}[m (\coth (mL)-1)]=-2\sum_{k\geq 1} \frac{k L}{\sqrt{4 \pi} t^{3/2}} e^{-\frac{k^2 L^2}{t}},
\end{equation}
 where the star stands for convolution in $t$. Therefore, the gluing formula (\ref{gluing formula}) becomes
 \begin{multline}\label{gluing of intervals 36}
       K^L(L_1+x,L_1+y|t) - K^{L_2}(x,y|t) = \sum_{n=0}^\infty\int_{t_0+\cdots+t_{n+2}=t,\, t_i> 0} dt_0\cdots dt_{n+1} \\
       \frac{1}{\sqrt{4\pi}t_0^{3/2}}\sum_{k_0\in \ZZ} (x+2k_0L_2) e^{-\frac{1}{4t_0}(x+2k_0 L_2)^2} \sum_{k_1,\ldots,k_n\geq 1} \prod_{j=1}^n\frac{k_j}{\sqrt{\pi} t_j^{3/2}} (L_1 e^{-\frac{k_j^2 L_1^2}{t_j}}+L_2 e^{-\frac{k_j^2 L_2^2}{t_j}})\cdot \\
       \cdot \frac{1}{\sqrt{4\pi t_{n+1}}} \cdot  \frac{1}{\sqrt{4\pi}t_{n+2}^{3/2}}\sum_{k_{n+2}\in \ZZ} (y+2k_{n+2}L_2) e^{-\frac{1}{4t_{n+2}}(y+2k_{n+2} L_2)^2}.
 \end{multline}

\begin{remark}
 The sum over $n$ in (\ref{gluing of intervals 36}) is absolutely convergent by the following argument. Denote minus the r.h.s. of (\ref{eq35}) by $\Phi^{L_1}(t)$. We are summing up $n$-fold convolutions of $\Phi=\Phi^{L_1}+\Phi^{L_2}$ which is a positive bounded function of $t$ (it is smooth, behaves as\footnote{Notice that as $t\to \infty$ one can bound the sum by
$$ \sum_{k\geq 1} kL\ e^{-\frac{k^2L^2}{t}} \leq \int_0^\infty kL\ e^{-\frac{k^2L^2}{t}} dk 
 = \frac{t}{2L}. $$ For $t\to 0$, we use $k^2 \geq k$ to bound 
 $$ \sum_{k\geq 1} kL\ e^{-\frac{k^2L^2}{t}} \leq \sum_{k\geq 1} kL (e^{-\frac{L^2}{t}})^k = L\frac{e^{-\frac{L^2}{t}}}{\left(1 -e^{-\frac{L^2}{t}}\right)^2 } \sim_{t \to 0} O\left(e^{-\frac{L^2}{t}}\right).$$  } $O\left(t^{-3/2} e^{-\frac{\mr{min}(L_1,L_2)^2}{t}}\right)$ at $t\ra 0$ and as $O(t^{-1/2})$ at $t\ra +\infty$). Hence, its $n$-fold convolution is positive and bounded by $C^n t^{n-1}/(n-1)!$ where $C=\mr{sup}_{t>0} \Phi(t)$ and $t^{n-1}/(n-1)!$ is the volume of $(n-1)$-simplex of size $t$. Hence, the sum over $n$ is absolutely convergent.

 On the other hand, the geometric progression (\ref{eq: geom progression for DN^-1}) in this case is only convergent for $m$ large enough, such that $|\coth (mL_1)+\coth(mL_2)-2|<2$.
 \end{remark}
 \subsection{Half-spaces}
The gluing formula for half-spaces $\RR_{\geq 0} \times \RR^{n-1}$ follows from the gluing of rays and the product and convolution properties of the heat kernel. {
Note that in this example, $\gamma = \RR^{n-1}$ is not compact. }
Namely, let 
\begin{equation*}
    K^n(x,y|t) = \frac{1}{(4\pi t)^{n/2}}e^{- \frac{|x-y|^2}{4t}}
\end{equation*}
be the heat kernel on $\RR^n$ and 
\begin{equation*}
    K^{(1)}((x_1,x), (y_1,y)|t) = \frac{1}{(4\pi t)^{n/2}}\left(e^{\frac{-(x_1-y_1)^2}{4t}} - e^{\frac{-(x+y)^2}{4t}}\right)e^{-\frac{|x-y|^2}{4t}}
\end{equation*}
be the heat kernel on $\RR
_{\geq 0} \times \RR^{n-1}$. Its normal derivative at $y=0$ is 
\begin{equation*}
-\partial_{y_1}K^{(1)}((x_1,x),(y_1,y)|t)\rvert_{y=0} = \frac{-1}{2\sqrt{\pi}t^{3/2}}x_1 e^{-\frac{x_1^2}{4t}} K^{n-1}(x,y|t) .
\end{equation*}
The gluing formula \eqref{eq: gluing heat kernel 1} then becomes 
\begin{multline*}
    K^n((x_1,x),(y_1,y)|t) - K^{(1)}((x_1,x),(y_1,y)|t)= \\ =  \frac{1}{8\pi^{3/2}}\int_{t_0+t_1+t_2=t, t_i > 0}dt_0dt_1\frac{xy\, e^{-\frac{x^2}{4t_0}-\frac{y^2}{4t_2}}}{\sqrt{t_0^3t_1t_2^3}} \cdot \\
    \cdot \underbrace{\int_{u,v \in \RR^{n-1}}dudv\, K^{n-1}(x,u|t_0)K^{n-1}(u,v|t_1)K^{n-1}(v,y|t_2)}_{=K^{n-1}(x,y|t_0+t_1+t_2) \text{ by convolution}} \\ 
    = \frac{1}{\sqrt{4\pi t}}e^{-\frac{(x_1+y_1)^2}{4t}}K^{n-1}(x,y|t).
\end{multline*}
\subsection{Cylinders}
Let $M_1=I_1\times \gamma$ and $M_2=I_2\times \gamma$ be two cylinders, with $I_1=[0,L_1]$, $I_2=[L_1,L_1+L_2]$ two intervals and $\gamma$ a Riemannian manifold (
{not necessarily compact}). The glued cylinder is $M=M_1\cup_{\{L_1\}\times \gamma}M_2= I\times \gamma$, where 
$I=I_1\cup I_2=[0,L_1+L_2]$ is the glued interval.
Formula (\ref{gluing formula 2}), using the factorization of heat operators $K^{I\times M}(t)=K^I(t)\otimes K^\gamma(t)$ becomes
\begin{multline*}
    K^M(t)|_{M_a\times M_b}=K^I(t)|_{I_a\times I_b}\otimes K^\gamma(t)=
    \delta_{ab}K^{I_a}(t)\otimes K^\gamma(t)+\\
    +\sum_{n\geq 0} \left((-1)^n \int_{t_0+\cdots t_{2n+2}=t,\, t_i>0} \dd^n K^{I_a}(t_0)\prod_{i=0}^{n-1}\Big(\frac{1}{\sqrt{4\pi t_{2i+1}}}K^{'I}(t_{2i+2})\Big) \cdot \right.\\
    \left. \cdot \frac{1}{\sqrt{4\pi t_{2n+1}}} (\dd^n K^{I_b}(t_{2n+2}))^*\right)
    \otimes \underbrace{(K^\gamma(t_1)\cdots K^\gamma(t_{2n+2}))}_{K^\gamma(t)}.
\end{multline*}
Here:
\begin{itemize}
\item $a,b\in \{1,2\}$; 
\item $(\cdots)|_{M_a\times M_b}$ stands for the operator $C^\infty(M_b)\ra C^\infty(M_a)$ defined by the restriction of the kernel of the operator $(\cdots)$ to $M_a\times M_b$; 
\item $\dd^n K^{I_a}(t)\colon C^\infty(\{L_1\})\ra C^\infty(I_a)$ is the 
operator defined by the kernel (\ref{interval normal derivative}) on the interval $I_a$ and $(\dd^n K^{I_a}(t))^*\colon C^\infty(I_a)\ra C^\infty(\{L_1\})$ is its dual (transpose); 
\item
$\displaystyle
    K^{'I}(t)=\sum_{k\geq 1}\sum_{a=1}^2\frac{1}{\sqrt{\pi} t^{3/2}}
    \left(1-\frac{2 k^2 L_a^2}{t}\right) e^{-\frac{k^2 L_a^2}{t}} .
$
\end{itemize}

Thus, the gluing formula in the case of cylinders amounts to the gluing formula for intervals, tensored with the heat operator on the slice $\gamma$ (making use of the convolution property of the latter).


\appendix

\section{Motivation: Cutting and gluing in renormalized quantum field theory}\label{sec: motivation}
Our main motivation comes from heat kernel renormalization of perturbative quantum field theory. Namely, consider the massive scalar field theory on a Riemannian manifold $M$: the space of fields is $F_M = C^\infty(M)$ and the action functional is 
\begin{equation*}
    S_M(\phi) = \int_M \frac{1}{2}d\phi\wedge * d\phi + \frac{m^2}{2}\phi * \phi + * p(\phi)
\end{equation*}
with $*$ the Hodge star operator on $M$ and $p(\phi) = \sum_{k=3}^N \frac{1}{k!}p_k \phi^k$ a polynomial, $m >0$ is a parameter 
(interpreted as mass). We are interested in the cutting and gluing behaviour of its perturbative partition function, given by
\begin{equation*}
        Z^{pert}_M(\hbar) \text{``=''} \lim_{\hbar\to 0} \int_{F_M}e^{-\frac{1}{\hbar}S_M(\phi)}D\phi := \frac{1}{\det_\zeta^{1/2}(\Delta + m^2)}\sum_{\Gamma}\frac{\hbar^{-\chi(\Gamma)
        }}{|\mr{Aut}(\Gamma)|}F(\Gamma) . 
\end{equation*}
Here in the right hand side
\begin{itemize}
    \item $\Delta$ denotes the positive Laplace operator on $M$ and $\det_\zeta$ the zeta-regularized determinant, 
    \item $\Gamma$ runs over all graphs which are at least trivalent and at most $N$-valent, 
    \item For a graph $\Gamma$: 
    \begin{itemize}
        \item $\chi(\Gamma)=|V_\Gamma|-|E_\Gamma| \leq 0$ is its Euler characteristic,
        \item $\mr{Aut}(\Gamma)$ its automorphism group, 
        \item $F(\Gamma)$ denotes its Feynman weight, given by 
        \begin{equation}
            F(\Gamma) = \int_{(x_1,\ldots x_{|V_\Gamma|})\in M^{|V_\Gamma|}}\prod_{v\in V(\Gamma)}(-p_{val(v)}) \prod_{e=(v_i,v_j) \in E(\Gamma)} G(x_i,x_j) dx_1\cdots dx_j. \label{eq: Feynman}
        \end{equation}
    \end{itemize}
\end{itemize}
Here $G \in C^\infty(M \times M \setminus \mr{diag})$ denotes the Green's function of the Helmholtz operator $\Delta + m^2$. Close to the diagonal, the Green's function behaves like 
\begin{equation*}
G(x,y) \simeq_{d(x,y) \to 0}     \begin{cases} O(1), & \dim M = 1, \\ 
C_2 \log d(x,y), &\dim M = 2, \\
C_{\dim M}\frac{1}{d(x,y)^{\dim M -2}}, & \dim M > 2 ,\\
    \end{cases} 
\end{equation*}
which means that some of the integrals given by \eqref{eq: Feynman} diverge if $\dim M > 1$. In dimension 2, the only divergent integrals are the ones from graphs containing an edge starting and ending at the same vertex (a short loop) as it naively leads to the evaluation of $G$ on the diagonal (where it is singluar). In \cite{KMW}, we showed how to regularize these in such a way that the result is compatible with cutting and gluing of $M$. However, for $\dim M \geq 3$ there are nontrivial divergent graphs and a variety of ways exist to make sense of them. For instance, if $K(x,y|t)$ denotes the heat kernel of the Laplace-Beltrami operator $\Delta$, then a naive solution to the problem of divergencies is to introduce a regulator $\epsilon$ and define 
\begin{equation*}
    G^\epsilon(x,y) = \int_\epsilon^\infty e^{-m^2 t}K(x,y|t) dt. 
\end{equation*}
Then $G^\epsilon \in C^\infty(M \times M)$, and thus at least on compact manifolds the regularized integrals $F^\epsilon(\Gamma)$, where we replace $G$ with $G^\epsilon$, are convergent. The next problem is to analyze the behaviour of the $F^\epsilon(\Gamma)$ as $\epsilon \to 0$, trying to consistently remove the divergencies using a finite amount of choices -- this leads to some restrictions on $\dim M$ and the interaction polynomial $p$ as to when this is possible.\footnote{A relatively recent write-up aimed a mathematical audience is the textbook by Costello \cite{Costello}. }

In \cite{KMW} we proved a gluing formula for the perturbative partition function of the form $\langle Z^{pert}_{M_1},Z^{\pert}_{M_2}\rangle = Z^{pert}_{M_1 \cup_\gamma M_2}$ in the case where $\dim M_i =2$. The eventual goal is to prove a similar formula in higher dimensions that is compatible with renormalization. The first step in this direction is to prove a gluing formula for the regularized partition functions
\begin{equation*}
    Z^{\pert, \epsilon}_M = \frac{1}{\det_\zeta(\Delta + m^2)^\frac12}\sum_{\Gamma}\frac{
    \hbar^{-\chi(\Gamma)}
    }{|\mr{Aut}(\Gamma)|}F^\epsilon(\Gamma).
\end{equation*}
A critical first step, covered in this paper, is to understand the behaviour of the heat kernel under gluing. 

\section{Proof of Proposition \ref{prop: kernel DN inverse}}\label{app: Proof of Prop}
We want to prove that for a Riemannian manifold decomposed as $M = M_1 \cup_\gamma M_2$ we have for $x,y \in \gamma$ that
$$(\DD^{\gamma,M}_{m^2})^{-1}(x,y) = G_{m^2}(x,y)$$ 
with $G_{m^2}$ denoting the Green's function on $M$. \\
Let $\eta \in C^\infty(\gamma)$, and define 
$$u(x) = \int_{y \in \gamma}G_{m^2}(x,y)\eta(y)dy.$$
Then $u \in L^2(M)$ is a smooth  function on $M \setminus \gamma$, and satisfies $(\Delta + m^2)u(x) = 0 $ for $x \in M \setminus \gamma$. For any function $\phi \in C^\infty(M)$ we have 
$$\int_M u(\Delta + m^2)\phi\, dx  = \int_{y\in \gamma} \eta(y) \int_{x \in M} G_{m^2}(y,x)(\Delta+m^2)\phi(x) dx dy = \int_\gamma \phi \eta\, dy,$$
by the definition of the Green's function. On the other hand, we have 
\begin{align*}
\int_M u (\Delta + m^2)\phi dx &= \int_{M_1} u (\Delta + m^2)\phi\, dx + \int_{M_2} u (\Delta + m^2)\phi\, dx \\ 
&= \int_{M_1} \left(u (\Delta + m^2)\phi - \phi (\Delta + m^2) u\right) dx \\
&+ \int_{M_2} \left(u (\Delta + m^2)\phi - \phi (\Delta + m^2) u\right) dx \\
&= \int_{M_1} \left(u \Delta \phi - \phi \Delta u \right)dx + \int_{M_2} \left(u \Delta \phi - \phi \Delta u\right) dx
\end{align*}

where in the second equality we have used that $u$ is $(\Delta + m^2)$-harmonic on ${M}^\circ_1 \sqcup M^\circ_2$. 
Now, using Green's formula we can rewrite this as 
\begin{multline*}
    \int_{M_1} (u \Delta \phi - \phi \Delta u) dx + \int_{M_2} (u \Delta \phi - \phi \Delta u) dx =\\
    =
    \int_\gamma (-u \partial_{M_1}^n \phi + \phi \partial^n_{M_1} u - u \partial^n_{M_2} \phi + \phi \partial^n_{M_2}) dy. 
\end{multline*}
Since $\phi \in C^\infty(M)$, it is in particular smooth across $\gamma$ and $\partial^n_{M_1}\phi + \partial^n_{M_2}\phi = 0$. Therefore, we finally arrive at 
$$\int_{\gamma} \phi \eta  \, dy = \int_M u (\Delta + m^2)\phi\, dy = \int_\gamma \phi \DD^{\gamma,M}_{m^2}u\, dy, $$ 
from which the statement follows. 

\section{Proof of Lemma \ref{lemma: Assump for product metric}} \label{appendix: proof of Lemma}
Let us represent $M$ as a gluing of a cylindrical collar of the boundary (with product metric), $\mr{cyl}=\gamma\times [0,\epsilon]$ (with the boundary components $\gamma=\gamma\times\{0\}$ and $\gp=\gamma\times\{\epsilon\}$),
and the complement $\Mp\subset M$:
$$M= 
\mr{cyl}\cup_{\gp}
\Mp,$$
see Figure \ref{fig: App A}.
\begin{figure}[h]
    \centering
    \includegraphics[scale=0.8]{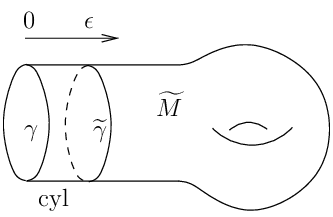}
    \caption{Decomposition of $M$ into a  cylindrical collar of the boundary and a complement $\til{M}$}
    \label{fig: App A}
\end{figure}
By the gluing formula for DN operators (which follows from the gluing formula (\ref{eq: gluing green's functions}) for Green's functions by taking the second normal derivative at the boundary), one has
\begin{equation}\label{DN gluing}
    \DN^{\gamma,M}= \DN^\mr{cyl}_{\gamma,\gamma} + \DN^\mr{cyl}_{\gamma,\gp} (\DD^{\gp,M})^{-1}  \DN^\mr{cyl}_{\gp,\gamma} .
\end{equation}
Here the subscripts $\gamma,\gp$ refer to a block of the $\DN$ operator; $\DD^{\gp,M}=\DN^{\mr{cyl}}_{\gp,\gp}+\DN^{\gp,\Mp}$ is the total $\DN$ operator on the interface $\gp$, as in (\ref{DN total}). We understand that all $\DN$ operators depend on $m^2$ and we temporarily suppress the $m^2$ subscripts. Subtracting $\sqrt{\Delta^\gamma+m^2}$ from both sides of (\ref{DN gluing}), we have
\begin{equation}\label{DN' gluing}
     (\DN^{\gamma,M})'= (\DN^\mr{cyl}_{\gamma,\gamma})' + \DN^\mr{cyl}_{\gamma,\gp} (\DD^{\gp,M})^{-1}  \DN^\mr{cyl}_{\gp,\gamma} .
\end{equation}
All the operators in the r.h.s. are bounded uniformly in $m$, for $\re m^2>c$ for any $c>0$:
\begin{itemize}
    \item $(\DN^\mr{cyl}_{\gamma,\gamma})'$ has eigenvalues $\mu_k (\coth(\epsilon \mu_k)-1)$ in the notations of Example \ref{example: cylinder}, which implies boundedness uniformly in $m$.
    \item Similarly, $(\DN^\mr{cyl}_{\gamma,\gp})'$ and $(\DN^\mr{cyl}_{\gp,\gamma})'$ have eigenvalues $-\frac{\mu_k}{\sinh(\epsilon\mu_k)}$ which are also bounded uniformly in $m$.\footnote{More explicitly, one has $|\mu_k (\coth(\epsilon \mu_k)-1)|, |\frac{\mu_k}{\sinh(\epsilon \mu_k)}| < \frac{1}{\epsilon}$ for any $m$ with $\re m^2> 0$ and any $k=0,1,2,\ldots$, which implies $|| (\DN^\mr{cyl}_{\gamma,\gamma})'||, ||(\DN^\mr{cyl}_{\gamma,\gp})' ||, || (\DN^\mr{cyl}_{\gp,\gamma})'|| <\frac{1}{\epsilon}$ for $\re m^2>0$.}
    \item The integral kernel  of the interface $\DN$ operator $(\DD^{\gp,M})^{-1}$ is the restriction of the Green's function on the total manifold $M$ to pairs of points on the interface $\gp$.\footnote{See e.g. the proof of Theorem 2.1 in \cite{Carron}.} Using this and the expression for the Green's function on $M$ via heat kernel, we have the equality of operators
    \begin{equation}\label{DNinv via kappa_t}
        (\DD^{\gp,M})^{-1}=\int_0^\infty dt\, e^{-m^2 t} \kappa_t,
    \end{equation}
    where $\kappa_t\colon C^\infty(\gp)\ra C^\infty (\gp)$ is the operator defined by the integral kernel -- the heat kernel on $M$, $K^M(x,y|t)$ restricted to $(x,y)\in \gp\times\gp$. Taking the operator norm of both sides of (\ref{DNinv via kappa_t}), we have
    \begin{equation}\label{DNinv bound}
        ||(\DD^{\gp,M})^{-1}|| \leq \int_0^\infty dt\, e^{-(\re m^2)t} || \kappa_t||.
    \end{equation}
    For $t>0$, the operator $\kappa_t$ is smoothing (in particular, bounded) and its norm behaves as 
    $O(t^{-\frac12 \dim M})$
    as $t\ra \infty$.\footnote{This follows from bounding the norm by the $l^1$ norm, $||\kappa_t||\leq \mr{tr}\, \kappa_t=\int_{\gp}K^M(x,x|t) dx=O(t^{-\frac{\dim M}{2}})$, see 
    e.g. \cite[f-la (1.1)]{Grigoryan}.
    } 
    For $t\ra 0$, one has the asymptotics\footnote{This follows from the short-time asymptotics of the heat kernel. In the asymptotic regime $t\ra 0$, one can replace $M$ by a cylinder $\gp\times \RR$; then one uses the product property of heat kernels for Cartesian products. In Appendix \ref{appendix A.1} below we give a detailed argument for the norm estimate (\ref{||kappa_t||}) that we need.}
    \begin{equation*}
    \kappa_t \sim \underbrace{(4\pi t)^{-1/2}}_{K^\RR(0,0|t)} K^{\gp}(t)\sim (4\pi t)^{-1/2}\mr{Id}.
    \end{equation*}
    Therefore, for the norm of $\kappa_t$ one has 
    \begin{equation}\label{||kappa_t||}
    ||\kappa_t||=O(t^{-1/2})
    \end{equation} 
    as $t\ra 0.$  This discussion implies that the r.h.s. of (\ref{DNinv bound}) is a monotonously decreasing positive function of $\re m^2$, which implies that the operator $(\DD^{\gp,M})^{-1}$ is bounded uniformly in $m$, for $\re m^2>c$ for any $c>0$.

\end{itemize}
Thus, the l.h.s. of (\ref{DN' gluing}) is bounded uniformly in $m$ for $\re m^2>c$, which immediately implies that Assumption \ref{conj: bound} holds.

\subsection{Proof of the estimate (\ref{||kappa_t||}).}\label{appendix A.1}
    For completeness, we give another, more detailed, proof of the estimate (\ref{||kappa_t||}). Without loss of generality, we can assume that $\gp$ has a neighborhood with product metric in $M$ (one always can arrange this by shifting $\gp$ toward $\gamma$). So, we have $M=\mr{cyl}\cup_{\gp}\Mp$. Let us iterate the construction and represent $\Mp$ as $\Mp=\til{\mr{cyl}}\cup_{\til{\gp}} \til{\Mp}$, i.e., we chop off a cylinder $\til{\mr{cyl}}=\gamma\times [\epsilon,\til{\epsilon}]$, with $\til{\epsilon}-\epsilon>0$ (and with boundaries $\gp=\gamma\times\{\epsilon\}$ and $\til{\gp}=\gamma\times \til{\epsilon}$ -- copies of $\gamma$), off $\Mp$ and call the remainder $\til{\Mp}$, see Figure \ref{fig: App A1}. 
    \begin{figure}[h]
        \centering
        \includegraphics[scale=0.8]{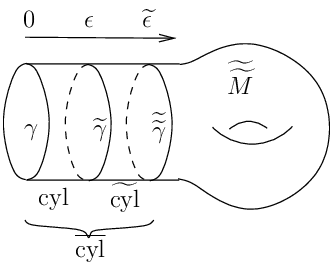}
       \caption{Removing of a collar from $M$, iterated twice.}
        \label{fig: App A1}
    \end{figure}
    
    Denote $\cylbar\colon= \cyl\cup_{\gp}\cylp$.
    By the formula (\ref{eq: gluing heat kernels 2}) applied to the decomposition $M=\cylbar\cup_{\gpp}\Mpp$, for the integral kernel of the operator $\kappa_t$ we have 
    \begin{multline}\label{A1 computation}
        \kappa_t(x,y)=K^M(x,y|t)=K^{\cylbar}(x,y|t)+\\
        +\int_{t_0+t_1+t_2=t, t_i\geq 0} dt_0 dt_1
        \int_{\gpp\times \gpp} dudv \dd^n_u K^{\cylbar}(x,u|t_0) K^M(u,v|t_1) \dd^n_v K^{\cylbar}(v,y|t_2)\\
        =K^\gamma(x,y|t) K^{[0,\ep]}(\epsilon,\epsilon|t)
        +\int_{t_0+t_1+t_2=t, t_i> 0} dt_0 dt_1 \\
        \int_{\gpp\times \gpp} dudv\  \left.\frac{d}{d\alpha}\right|_{\alpha=\ep} K^{[0,\ep]}(\epsilon,\alpha|t_0)  K^{\gamma}(x,u|t_0) K^M(u,v|t_1)    \left. \frac{d}{d\beta}\right|_{\beta=\ep} K^{[0,\ep]}(\beta,\epsilon) K^{\gamma}(v,y|t_2)
    \end{multline}
    for $x,y\in \gp$. In the second step above we used the factorization property of heat kernel on a Cartesian product. Going back from integral kernels to operators, (\ref{A1 computation}) reads
    \begin{multline}\label{A1 computation 2}
        \kappa_t= K^{[0,\ep]}(\epsilon,\epsilon|t) K^\gamma(t)  +\\+\int_{t_0+t_1+t_2=t,t_i> 0} dt_0 dt_1 
        \left.\frac{d}{d\alpha}\right|_{\alpha=\ep} K^{[0,\ep]}(\epsilon,\alpha|t_0)  K^{\gamma}(t_0) \kappa_{t_1}^{\gpp}   K^{\gamma}(t_2)  \left. \frac{d}{d\beta}\right|_{\beta=\ep} K^{[0,\ep]}(\beta,\epsilon), 
    \end{multline}
    where $\kappa_{t}^{\gpp}\colon C^\infty(\gpp)\ra C^\infty(\gpp)$ is the operator defined by the integral --  the restriction of the heat kernel on $M$ to pairs of points on $\gpp$; $K^\gamma(t)\colon C^\infty(\gamma)\ra C^\infty(\gamma)$ is the heat operator on $\gamma$. Taking the operator norm of both sides of (\ref{A1 computation 2}), we have
    \begin{multline}\label{A1 computation 3}
         ||\kappa_t|| \leq  K^{[0,\ep]}(\epsilon,\epsilon|t) ||K^\gamma(t)||  +\int_{t_0+t_1+t_2=t,t_i> 0} dt_0 dt_1\\ 
        \left.\frac{d}{d\alpha}\right|_{\alpha=\ep} K^{[0,\ep]}(\epsilon,\alpha|t_0) \cdot ||K^{\gamma}(t_0)||\cdot ||\kappa_{t_1}^{\gpp}||\cdot   ||K^{\gamma}(t_2)||\cdot  \left. \frac{d}{d\beta}\right|_{\beta=\ep} K^{[0,\ep]}(\beta,\epsilon).
    \end{multline}
    Next, we have the following:
    \begin{enumerate}[(a)]
        \item  $||K^\gamma(t)||=||e^{-t\Delta^\gamma}||=1$ for any $t>0$ (since zero is the lowest eigenvalue of $\Delta^\gamma$).
        \item 
            Bounding the operator norm of $\kappa^{\gpp}_{t}$ by the $l^p$ norm, with any $p\geq 1$, one has
        \begin{multline}\label{A1 computation 5}
            ||  \kappa^{\gpp}_t || \leq \left(\mr{tr}(\kappa^{\gpp}_t)^p\right)^{1/p}=\\
            =
            \left(\int_{\gpp^{\times p}} dx_1\cdots dx_p K^M(x_1,x_p|t)K^M(x_p,x_{p-1}|t)\cdots K^M(x_2,x_1|t) \right)^{1/p}
            \\
            = \left(\int_{\gpp}dx_1 \int_{{\gpp}^{p-1}}  dx_2\cdots dx_p K^M(x_1,x_p|t)K^M(x_p,x_{p-1}|t)\cdots K^M(x_2,x_1|t)\right)^{1/p}\\
            = O((t^{-\frac{\dim M}{2}p}t^{\frac{\dim M-1}{2}(p-1)})^{1/p})=O(t^{-\frac12-\frac{\dim M-1}{2p}}),
        \end{multline}
        where we in the last step use the short-time asymptotics of the heat kernel, see e.g. \cite[Theorem 2.30]{BGV}.\footnote{
In the asymptotic regime $t\ra 0$, the relevant configurations of points for the integral in (\ref{A1 computation 5}) are when $x_1,\ldots,x_p$ are close together (with pairwise distances $O(\sqrt{t})$). So, the asymptotics is given by integrating, say, $x_1$ over $\gpp$, and for the integration over the remaining points, one can asymptotically replace the heat kernel on $M$ by the one on flat space $\RR^{\dim M}$. In the r.h.s. of (\ref{A1 computation 5}), the $O(t^{\frac{\dim M-1}{2}(p-1)})$ is  the volume of the contributing region in the integral over $x_2,\ldots,x_p$ for fixed $x_1$; $O(t^{-\frac{\dim M}{2}p})$ is the maximum of the integrand.
        } Hence, for any $\delta>0$ one has
        \begin{equation}\label{||kappa|| L^p bound}
            ||\kappa^{\gpp}_t ||=O(t^{-1/2-\delta}),
        \end{equation}
        as follows from selecting large enough $p$ in the estimate above.
        \item From the explicit formula for the heat kernel on an interval (\ref{K interval}), we have the bounds 
        \begin{equation*}
            K^{[0,\ep]}(\epsilon,\epsilon|t)<(4\pi t)^{-1/2}+e^{-c_1/t},\qquad
           \left| \left.\frac{d}{d\alpha}\right|_{\alpha=\ep} K^{[0,\ep]}(\epsilon,\alpha|t)\right|< e^{-c_2/t}
        \end{equation*}
        for $t$ sufficiently small, with some constants $c_1,c_2>0$.
    \end{enumerate}
    Therefore, (\ref{A1 computation 3}) implies, for $t$ sufficiently small,
    \begin{multline}\label{A1 computation 4}
        ||\kappa_t|| \leq (4\pi t)^{-1/2}+
        \underbrace{
        e^{-c_1/t}+c_3\int_{t_0+t_1+t_2=t,t_i>0} e^{-c_2/t_0-c_2/t_2} 
        t_1^{-\frac{1}{2}-\delta}
        }_{O(e^{-c_4/t})},
    \end{multline}
    with $c_{1,\ldots,4}$ some positive constants. Here we assumed that we fixed some $\delta<\frac12$ in (\ref{||kappa|| L^p bound}) (which is important for convergence of the integral in (\ref{A1 computation 4}) in the region $t_1\ra 0$).
    Hence we have the desired estimate $||\kappa_t||=O(t^{-1/2})$ for $t\ra 0$.

\section{Examples of Dirichlet-to-Neumann operators (for Assumption \ref{conj: bound})}
\label{appendix: DN examples for Conj}
\begin{example}[Disk]
Let us consider a disk of radius $R$. Then $\gamma = S^1$ and an eigenbasis of $\Delta^\gamma$ is given by $\eta_k(\theta) = e^{ik\theta}$ for $k \in \ZZ$ with eigenvalues $\omega_k = k^2/R^2$. The harmonic extension to the inside of the disk is given by $$\phi_k(r,\theta) = \frac{I_k(mr)}{I_k(mR)}e^{ik\theta}$$ with $I_n(x)$ the modified Bessel function solving 
$$ x^2I_n''(x) + xI_n'(x) - (x^2 + n^2)I_n(x) = 0.$$ Therefore, the eigenvalues of the Dirichlet-to-Neumann operator are 
\begin{equation*}
    \lambda_k = m\frac{I'_k(mR)}{I_k(mR)}.
\end{equation*}
One can show\footnote{E.g. from the series expansion of $I_n(z)$ one has for $k \to \infty$ that $$\lambda_k = \frac{k}{R}\left(1 + \frac{m^2R^2}{2k^2} + O(k^{-3})\right).$$ On the other hand, we have $k \to \infty$ $$\mu_k = \sqrt{m^2 + \frac{k^2}{R^2}} = \frac{k}{R}\sqrt{1+\frac{(mR)^2}{k^2}} = \frac{k}{R}\left(1 + \frac{1}{2}\frac{m^2R^2}{k^2} - \frac{1}{8}\frac{(mR)^4}{k^4}+ O(k^{-6})\right).$$}
that for fixed $m$, as $k \to \infty$,
\begin{equation*}
    \lambda_k - \mu_k =  m\frac{I'_k(mR)}{I_k(mR)} - \sqrt{m^2+\frac{k^2}{R^2}} =  O(k^{-2})
\end{equation*}
and therefore the operator $\DN^{S^1_R,D_R}_{m^2} - \sqrt{\Delta_{S^1}+m^2}$ is bounded. To prove the bound on the norm we notice that for large $m$ and fixed $k$ the eigenvalues behave as\footnote{This follows from the large $z$ asymptotics $I_k(z) \sim \frac{e^z}{\sqrt{2\pi z}}\left(1 + O(z^{-1})\right)$ which imply +$I_n'(z)/I_n(z) \sim 1 - \frac{1}{2z} + O(z^{-2})$, and $\sqrt{m^2 + \frac{k^2}{R^2}} = m + O(m^{-1})$.}
\begin{equation*}
    \lambda_k - \mu_k = \frac{-1}{2R} + O(m^{-1})
\end{equation*}
and therefore we have 
\begin{equation*}
    \frac{1}{m}(\lambda_k - \mu_k) = \frac{-1}{2mR} + O(m^{-2}),
\end{equation*}
which implies Assumption \ref{conj: bound}.
\end{example}
\begin{example}[Spherical sector]
We let $\Sigma_{\varphi_0,R}$ be a spherical sector of angle $0 < \varphi_0 < \pi$ (for $\pi/2$ we have a hemisphere) on a sphere of radius $R$. The boundary is a circle of radius $R\sin \varphi_0$. By separation of variables one can show that the unique harmonic extension of $\eta_k = e^{ik\theta} \in
L^2(S^1_{R\sin\varphi_0})$ is the function
$$\phi_k(\varphi,\theta) = \frac{P^k_{\alpha}(\cos\varphi)}{P^k_{\alpha}(\cos\varphi_0)}\eta_k(\theta),$$ 
where $P^\mu_\nu(x)$ is the Legendre function and $\alpha$ either one of the roots 
of $\alpha^2 + \alpha + (mR)^2 = 0$. \footnote{In spherical coordinates, the Helmholtz equation $(\Delta_{S^2} + m^2)f(\varphi,\theta) = 0$ in the product ansatz $f = F(\varphi)e^{ik\theta}$ becomes $\varphi F''(\varphi) + \cot\varphi F'(\varphi) - \left(\frac{k^2}{\sin^2\varphi}+m^2R^2\right)F(\varphi) = 0$. The solution which is regular at 
$\varphi=0$
is the associated Legendre function of the first kind $F(\phi)= P^k_{\alpha}(\cos\varphi)$ with $\alpha^2 + \alpha  = - (mR)^2.$   } 
In particular, the Dirichlet-to-Neumann operator is diagonal in that basis and the eigenvalues are
\begin{equation*}
    \lambda_k = R^{-1} \left.\frac{d}{d\varphi}\right|_{\varphi=\varphi_0}\log P^k_{\alpha}(\cos\varphi). 
\end{equation*}
Then one can show $\lambda_k - \mu_k = O(k^{-2})$ as $k \to \infty$ for fixed $m$ (see \cite[Eq. (A.20)]{KMW}), which shows that $(\DN^{S^1_R,\Sigma_{\varphi_0,R}}_{m^2})'$ is a bounded operator. 

 For $m^2$ real, assumption in this example holds by Remark \ref{rem: Conj when DN and A commute} (since the spherical sector is a surface of revolution).

\end{example}
\section{Path integral interpretation of the gluing formula}\label{app: path integral}

    Recall from Remark \ref{rem: path integral} that an alternative way to write the gluing formulae for the Green's function and the heat kernel is 
    \begin{equation}
        G_{m^2}(x,y) = \delta_{ab}G^{(a)}_{m^2}(x,y) + \int_{\gamma \times \gamma} \partial^n_uG^{(a)}_{m^2}(x,u)G_{m^2}(u,v)\partial^n_vG^{(b)}_{m^2}(v,y) du dv  \label{eq: gluing greens functions alt}
    \end{equation} 
    and 
    \begin{multline}
         K(x,y|t) = \delta_{ab}K^{(a)}(x,y|t) +  \\                                          
    \int_{t_0+t_1+t_2 = t, t_i > 0}dt_0dt_1\, \int_{\gamma \times \gamma} dudv\, \partial^n_uK^{(a)}(x,u|t_0)K(u,v|t_1)\partial^n_vK^{(b)}(v,y|t_2).\label{eq: gluing heat kernels path int}
    \end{multline}
Heuristically, these formulae are evident from the 
  path integral (Feynman-Kac) formulae for  the heat kernel and the Green's function respectively.\footnote{
 These formulae also arise in the context of first quantization, see e.g. the introduction in \cite{CKMW}.
}
For instance, the 
 path integral 
formula for the heat kernel is \cite{Feynman-Hibbs}
\begin{equation*}
    K_{m^2}(x,y|t) = \int_{P^t_M(x,y)}e^{-S_{m^2}(p)}\mathcal{D}p
\end{equation*}
with 
$$ P^t_M(x,y) = \{ p\colon [0,t]\to M, p(0) = x, p(t) = y \}$$ 
and  
\begin{equation*}
    S_{m^2}(p) = \int_0^t \left(\frac{\dot{p}^2}{4} +m^2 - \frac{R(p(\tau))}{6} \right)d\tau  = m^2t + \int_0^t \left(\frac{\dot{p}^2}{4} - \frac{R(p(\tau))}{6}\right)d\tau,
\end{equation*}
where $R$ is the scalar curvature and $\dot{p}^2(\tau) = g_{p(\tau)}(\dot{p},\dot{p})$.  More precisely, we interpret $e^{-S_{m^2}(p)}\mathcal{D}p$ as $e^{-m^2t}$ times the Wiener measure induced by the heat kernel on $P^t_M(x,y)$, see for instance \cite{AD}, \cite{BS}. Consider now formula \eqref{eq: gluing heat kernels path int} with $a=b=1$, i.e., we have $x,y\in M_1$. We can decompose the set
$$P^t_M(x,y) = \{ p\colon [0,t]\to M, p(0) = x, p(t) = y \} \equiv P^t_{M,\gamma}(x,y) \sqcup P^t_{M_1\setminus \gamma}(x,y)$$
into paths that touch the interface $\gamma$ and those that do not. Integrating over the latter gives $K^{(1)}(x,y|t)$,\footnote{See for instance \cite[Corollary 3.11]{BS}.} while a path $p$ that does touch $\gamma$ can be written uniquely as a concatenation of three paths $p = p_0 * p_1 * p_2$: Denoting $u$ (resp. $v$) the first (resp. last) intersection point of $p$ with $\gamma$, we can cut $p$ into a path $p_0$ of length $t_0$ from $x$ to $u$, a path $p_1$ of length $t_1$ from $u$ to $v$ and a third path $p_2$ from $v$ to $y$ of length $t_2 = t - t_0 - t_1$. Note that $p_0$, $p_2$ have the property that they only touch $\gamma$ at the end (resp. start) 
- denote the set of such paths by $(P_{M,\gamma}^{t_0})'(x,u)$, resp. $(P_{M,\gamma}^{t_2})'(v,y)$. {
The map $p \to (p_0,p_1,p_2)$ defines a bijection 
$$ P^t_{M,\gamma}(x,y) = \bigsqcup_{\Delta^2_t} \bigsqcup_{u,v \in \gamma} (P^{t_0}_{M,\gamma})'(x,u)\times P^{t_1}_M(u,v) \times (P^{t_2}_{M,\gamma})'(v,y), $$
where $\Delta^2_t$ denotes the 2-simplex $\{(t_0,t_1,t_2)| t_i >0, \sum_i t_i =t\}.$ 
Let us now assume that the fictional measure $\mathcal{D}p$ factorizes as 
\begin{equation}
\mathcal{D}p = dt_1 dt_2 \cdot du\,dv\cdot (\mathcal{D}p_0)'\mathcal{D}p_1 (\mathcal{D}p_2)' ,\label{eq: measure factor}
\end{equation} 
for appropriately defined measures $(\mathcal{D}p_i)'$.
}
 Since $S_{m^2}(p_0*p_1*p_2) = S_{m^2}(p_0) + S_{m^2}(p_1) + S_{m^2}(p_2)$, we then obtain 
\begin{multline}
     K_{m^2}(x,y|t) = \int_{P^t_M(x,y)}e^{-S_{m^2}(p)}\mathcal{D}p 
     = \int_{P^t_{M\setminus \gamma}(x,y)}e^{-S_{m^2}(p)}\mathcal{D}p\, + \int_{P^t_{M,\gamma}(x,y)}e^{-S_{m^2}(p)}\mathcal{D}p \\
     =  \int_{P^t_{M\setminus \gamma}(x,y)}e^{-S_{m^2}(p)}\mathcal{D}p 
     + \int_{t_0+t_1+t_2 = t, t_i > 0}dt_0dt_1\, \int_{\gamma \times \gamma} dudv \\\int_{(P^{t_0}_{M,\gamma})'(x,u)}e^{-S_{m^2}(p_0)}(\mathcal{D}p_0)'\int_{P^{t_1}_{M,\gamma}(u,v)}e^{-S_{m^2}(p_1)}\mathcal{D}p_1\int_{(P^{t_2}_{M,\gamma})'(v,y)}e^{-S_{m^2}(p_2)}(\mathcal{D}p_2)' .\label{eq: path integral formula gluing}
\end{multline}
Integrating over all possible times $t_i$ and $u,v$, we obtain the second term of \eqref{eq: gluing heat kernels path int} {
\begin{equation}
    \int_{(P^{t_0}_{M,\gamma})'(x,u)}e^{-S_{m^2}(p_0)}(\mathcal{D}p_0)' = \partial^n_u K^{(1)}(x,u|t_0), \label{eq: heat kernel derivative} 
\end{equation} which is suggested by the discrete case, see Appendix \ref{sec: discrete} below, in particular Section \ref{rem: ext operator}, and the discussion in the introduction of \cite{CKMW}.
{
It would be interesting to understand equations \eqref{eq: measure factor} and \eqref{eq: heat kernel derivative} from the measure-theoretic viewpoint, as properties of the Wiener measure on a Riemannian manifold. }


\section{A gluing formula for the heat kernel on a graph}\label{sec: discrete}
In this appendix, we prove a gluing formula for the heat kernel on a graph, and show that in this case the path integral formulae of the previous appendix obtain a precise meaning as path sum formulae. We will provide two proofs of the gluing formula: one by computing inverses of block matrices and one by counting paths. 
\subsection{Generalities} Let $X = (V_X,E_X)$ be a graph, which for simplicity we assume to be simple (i.e. no short loops or multiple edges) and finite. We write $Y \subset X$ for a full subgraph $Y = (V_Y,E_Y)$. 
 In the following we will consider operators on the space $C^0(X) \cong \RR^{V_X}$ of 0-cochains on $X$, those are matrices $A$ with row and columns labeled by $V_X$, the entries of such a matrix $A$ are denoted $A(u,v)$. 

For example, $D^X$ is the diagonal 
matrix with $D^X(u,v) = \delta_{uv}\mr{val}(v)$, called the degree or valency matrix. Let $A^X$ be the adjacency matrix, with entries $A^X(u,v) = 1$ if $(u,v) \in E_X$ and $A^X(u,v) = 0$ otherwise. The graph Laplacian is 
\begin{equation*}
    \Delta^X := D^X - A^X \colon C^0(X) \to C^0(X)
\end{equation*}
 The heat kernel on $X$ is the operator $K^X(t):=e^{-t\Delta^X}$, in analogy with the continuum case we denote its entries by $K^X(u,v|t)$.  
\subsection{Gluing formula for Green's functions}
Now, let $Y \subset X$ be a full subgraph.\footnote{The situation in the main body of the text corresponds to the case where $X \setminus Y = X_1 \sqcup X_2$ is disconnected.}  
Let $G^X_{m^2}$ be the inverse matrix of $\Delta^X + m^2$, then we have block decompositions\footnote{Here, $\Delta^{X,Y}\colon C^0(X\setminus Y) \to C^0(X \setminus Y)$ denotes the Laplacian operator on $X$ ``relative to $Y$'', it is defined as $(D^X - A^X)|_{X\setminus Y}$. In particular, it is different from $\Delta^{X \setminus Y} = {D^{X \setminus Y} - A^{X\setminus Y}} \colon {C^0(X\setminus Y)} \to C^0(X \setminus Y)$. In the latter $D^{X\setminus Y}$ contains valencies of vertices in $X\setminus Y$ whereas $D^X|_{X\setminus Y}$ contains valencies of vertices in $X$ (including edges which start at vertices in $X\setminus Y$ and end at vertices of $Y$).}
\begin{equation}
  \Delta^X + m^2  = \left( 
    \begin{array}{c|c}
         \wh{A}=\Delta^{X,Y} + m^2& \wh{B} \\ \hline
         \wh{C}& \wh{D}
    \end{array}
    \right), \;\;
  G^X_{m^2} = \left( 
    \begin{array}{c|c}
         A& B \\ \hline
         C& D
    \end{array}
    \right),  \label{eq: block decomposition}
\end{equation}
with the blocks corresponding to vertices of $X \setminus Y$ and $Y$ respectively. 
We have a Dirichlet problem for a function $\varphi\in C^0(X)$ 
with boundary condition
$\eta \in C^0(Y)$:
\begin{equation}
    \begin{cases}
        (\Delta^{X}+m^2 )\varphi(x) &= 0, \text{ for all } x \in X\setminus Y, \\
        \varphi(y) &\equiv \eta(y), \text{ for all } y \in Y,
    \end{cases} \label{eq: dirichlet problem discrete}
\end{equation}
which has a unique solution $\phi_\eta$ for any $\eta \in C^0(Y)$. We denote $E^{Y,X}_{m^2}\colon C^0(Y) \to C^0(X)$ the map that associates to $\eta \in C^0(Y)$ the unique solution $\varphi_\eta$ of the Dirichlet problem \eqref{eq: dirichlet problem discrete}. It is easy to show, see \cite{CKMW}, that 
\begin{align*}
    E^{Y,X}_{m^2} \colon C^0(Y) &\to C^0(X) \cong C^0(X\setminus Y) \oplus C^0(Y) \\
    \eta &\mapsto (BD^{-1}\eta,\eta)  
\end{align*}
i.e. $E^{Y,X}_{m^2} = BD^{-1} \oplus 1_Y$. We define the relative Green's function $$G^{X,Y}_{m^2} := \wh{A}^{-1} = (\Delta^{X,Y} + m^2)^{-1} \colon C^0(X\setminus Y) \to C^0(X \setminus Y),$$ and extend it to $C^0(X)$ by zero on $Y$.  Finally, we define $$\DD_{m^2}^{Y,X} := D^{-1}$$ 
to be the total combinatorial Dirichlet-to-Neumann operator.\footnote{In the situation where $X = X_1 \cup_Y X_2$, one can show that $\DD^{Y,X}_{m^2} = \DD^{Y,X_1}_{m^2} + \DD^{Y,X_2}_{m^2} - (\Delta^Y + m^2)$.} Then we have a gluing formula for the Green's function 
\begin{equation}
    G^X_{m^2} = G^{X,Y}_{m^2} + E^{Y,X}_{m^2} (\DD^{Y,X}_{m^2})^{-1} (E^{Y,X}_{m^2})^T,
    \label{eq: gluing Green's function discrete}
\end{equation}
\begin{remark}[Discrete cutting formula] 

   The gluing formula \eqref{eq: gluing Green's function discrete} follows from the Schur complement formula 
   \begin{equation}
     G^{X,Y}_{m^2} = \wh{A}^{-1} = A - BD^{-1}C = \underbrace{A}_{G_X|_{X \setminus Y}} - \underbrace{(BD^{-1})}_{(E^{Y,X}_{m^2})}\underbrace{D}_{(\DD^{Y,X}_{m^2})^{-1}}\underbrace{(BD^{-1})^T}_{(E^{Y,X}_{m^2})^T}. \label{eq: Schur complement}
   \end{equation}
   Another way to read \eqref{eq: Schur complement} is that, on $X\setminus Y$,  
   \begin{equation*}
       G^{X,Y}_{m^2} - G^X_{m^2}|_{X \setminus Y} = - B D^{-1} C = - G^X_{m^2}\big|_{(X\setminus Y) \times Y}\DD^{Y,X}_{m^2}G^X_{m^2}\big|_{Y \times (X\setminus Y)} 
   \end{equation*}
   -- the discrete analog of the cutting formula \eqref{eq: cutting}, which shows that these two formulae are equivalent in the discrete case.
\end{remark}
\subsubsection{Extension operator as a discrete normal derivative} \label{rem: ext operator} 
    In the continuous 
    case, the (integral kernel of the) extension operator is given by the normal derivative of the Green's function with Dirichlet boundary conditions, i.e. if $M$ is a Riemannian manifold with boundary $\gamma$, then for a function $\eta \in C^\infty(\gamma)$, the function $\varphi_\eta \in C^\infty(M)$ given by 
    \begin{equation*}
        \varphi_\eta(x) = \int_{\gamma}\partial^n_uG_{m^2}(x,u)\eta(u)du
    \end{equation*}
    satisfies $(\Delta^M + m^2)\varphi_\eta = 0$ and $\varphi_\eta\big|_\gamma = \eta$. We now want to see how this is reflected in the discrete setting. First, notice that another way to write the extension operator  $(E^{Y,X}_{m^2})$ is as\footnote{This follows from $\wh{A}B + \wh{B}D = 0$, which is immediate from $G^X_{m^2} = (\Delta^X + m^2)^{-1}$ and the block decomposition \eqref{eq: block decomposition}.}
    \begin{equation}
        E^{Y,X}_{m^2}\bigg|_{(X\setminus Y) \times Y} = BD^{-1} = - \wh{A}^{-1}\wh{B}, \label{eq: ext op alt}
    \end{equation}
    which can be interpreted as a discrete analog of the normal derivative of the Green's function as follows. The entries of the matrix $\wh{B}$ are labeled by a pair of vertices $(x,u)$ with $x \in X\setminus Y$ and $u \in Y$, and $\wh{B}(x,u) = -1$ precisely if $(x,u) \in E_X$ and $\wh{B}(x,u) = 0$ otherwise (that is $\wh{B} = - A^X\big|_{X\setminus Y \times Y}$). Therefore, 
    \begin{equation*}
        E^{Y,X}_{m^2}(x,u)\bigg|_{(X\setminus Y) \times Y} = \sum_{x'\in X \setminus Y\colon (x',u) \in E_X} G^{X,Y}_{m^2}(x,x'). 
    \end{equation*}
    By definition, $G^{X,Y}_{m^2}(x_1,x_2) = 0$ if either $x_1 \in Y$ or $x_2 \in Y$, therefore, we can write $E^{Y,X}_{m^2}$ as 
    \begin{equation}
         E^{Y,X}_{m^2}(x,u) = \left(\sum_{x'\in X \setminus Y\colon (x',u) \in E_X} G^{X,Y}_{m^2}(x,x') - G^{X,Y}_{m^2}(x,u)\right). \label{eq: combinatorial normal derivative}
    \end{equation}
    Equation \eqref{eq: combinatorial normal derivative}  is a combinatorial analog of a normal derivative in the sense of finite differences. For instance, if we suppose that $X$ is a line graph and $Y$ is one of the endpoints, then \eqref{eq: combinatorial normal derivative} is just a single difference and with correct normalization will approximate the derivative of $G$ in the continuum limit.  More generally, this interpretation holds for sequences of pairs $(X,Y)$ ``converging'' to a domain with boundary, with each vertex of $Y$ having exactly one adjacent vertex in $X\setminus Y$.
\subsection{Gluing formulae for the heat kernel}
\subsubsection{Discrete gluing formula I}
Let us consider the boundary value problem for the heat equation on $X$. For a function $\varphi\in C^\infty(\RR_{\geq 0}, C^0(X))$ and $\eta \in C^0(Y)$, it reads:
\begin{equation}
    \begin{cases}
        (\partial_t + \Delta^{X})\varphi(t,x) = 0,  &  t>0,\; x\in X \setminus Y; \\
         \varphi(0,x) = 0, & x\in X\setminus Y ;\\
        \varphi(t,y) \equiv \eta(y), & t \geq 0,\; y \in Y .
    \end{cases} \label{eq: dirichlet problem discrete 2}
\end{equation}
For any given $\eta \in C^0(Y)$, there is a unique $\varphi_\eta$ solving the Dirichlet problem \eqref{eq: dirichlet problem discrete 2}, which defines a map $\EE^{Y,X}\colon C^0(Y) \to C^\infty(\RR_{\geq 0}, C^0(X))$.  
 One can check that 
\begin{equation*}
    (\EE^{Y,X}\eta)(t,x) = \int_0^t dt' (\varepsilon^{Y,X} \eta)(t',x) ,
\end{equation*} where 
$\varepsilon^{Y,X}(t)\colon C^0(Y) \to C^0(X)$ is given by 
\begin{equation}
    (\varepsilon^{Y,X}\eta)(t, x) = \sum_{y \in Y}\LL^{-1}[E^{Y,X}_{m^2}(x,y)](t)\eta(y). \label{eq: extension op heat}
\end{equation}
Now, we can state and prove the discrete analog of the first gluing formula for the heat kernel on a graph: 
\begin{proposition}\label{prop: gluing heat kernel discrete}
    We have  
    \begin{equation}
       K^X(t) = K^{X,Y}(t) + \int_{\substack{t_1+t_2+t_3 = t \\ t_i >0}}dt_1dt_2\,\varepsilon^{Y,X}(t_1)\LL^{-1}[(\DD^{Y,X}_{m^2})^{-1}](t_2)(\varepsilon^{Y,X}(t_3))^T.\label{eq: gluing heat kernel discrete}
    \end{equation}
\end{proposition}

\begin{proof}
    As in the continuous case, one can get the gluing formula for the heat kernel \eqref{eq: gluing heat kernel discrete} from the gluing formula for the Green's function \eqref{eq: gluing Green's function discrete} by applying the inverse Laplace transform and using the fact that multiplication gets mapped to convolution. The inverse Laplace transform of the Green's functions are the corresponding heat kernels and the inverse Laplace transform of the extension operator $E^{Y,X}_{m^2}$ is $\varepsilon^{Y,X}(t)$ \eqref{eq: extension op heat}. 
\end{proof}
\begin{remark}\label{rem: gluing formula on interface}
    Equation \eqref{eq: gluing heat kernel discrete} is valid on all of $Y$, if we define $K^{X,Y}(u,v|t) = 0$ if $u \in Y$ or $v \in Y$ and $\varepsilon^{Y,X}(x,y|t) = \delta(t)$ for $x \in Y$.
\end{remark}
\begin{example}
    Consider $X$ a line graph with 3 vertices.\footnote{
     A more common name for it is ``path graph.''  We use the term ``line graph'' instead, to avoid confusion with paths on a graph $X$.
    } Then 
    \begin{gather*}\Delta^X = \begin{pmatrix}
        1 & -1 & 0 \\ 
        -1 & 2 & -1 \\ 
        0 & -1 & 1
    \end{pmatrix},
    \\ K^X(t) = e^{-t\Delta^X} = \frac16\left(
\begin{array}{ccc}
 e^{-3 t}+3 e^{-t}+2 & 2-2 e^{-3 t} & e^{-3 t}-3 e^{-t}+2 \\
 2-2 e^{-3 t} & 4 e^{-3 t}+2 & 2-2 e^{-3 t} \\
 e^{-3 t}-3 e^{-t}+2 & 2-2 e^{-3 t} & e^{-3 t}+3 e^{-t}+2 \\
\end{array}
\right).
\end{gather*}
Let us glue $X = X_1 \cup_Y X_2$, with $X_i$ line graphs with 2 vertices and $Y$ a single vertex in $X_i$, see Figure \ref{fig: line graph gluing}. 
\begin{figure}[h]
    \centering
    \includegraphics[scale=0.8]{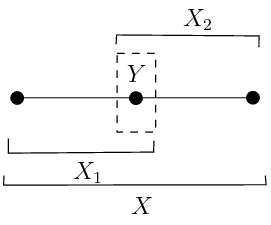}
    \caption{Line graph as a gluing of two subgraphs.}
    \label{fig: line graph gluing}
\end{figure}
Then 
$$\Delta_{X_1} = \begin{pmatrix}
    1 & -1 \\ 
    -1 & 1
\end{pmatrix},\quad  G_{X_1,m^2} = \frac{1}{m^2(2+m^2)}\left(
    \begin{array}{c|c}
        1+m^2 & 1 \\ \hline
        1 & 1+m^2
    \end{array}
    \right).
    $$
    The extension operator, and its inverse Laplace transform, are 
    \begin{align*} E^{Y,X_1} &= (BD^{-1}, 1) = \left( \frac{1}{1+m^2},1 \right), \\ \varepsilon^{Y,X_1}(t) &= \left(\LL^{-1}\left[\frac{1}{1+m^2}\right],1\right) = (e^{-t},\delta(t)).
    \end{align*}
    The total Dirichlet-to-Neumann operator is 
    \begin{multline*}
    \DD^{Y,X}_{m^2} = \DD^{Y,X_1}_{m^2}+\DD^{Y,X_2}_{m^2}-(\Delta^Y + m^2) =\\
    =\frac{m^2(2+m^2)}{1+m^2}+\frac{m^2(2+m^2)}{1+m^2}-m^2 = \frac{m^2(3+m^2)}{1+m^2} .
    \end{multline*}
    And the inverse Laplace transform of its inverse
    is 
    $$(\DD^{Y,X}_{m^2})^{-1} = \frac{1+m^2}{m^2(3+m^2)},\qquad \LL^{-1}[(\DD^{Y,X}_{m^2})^{-1}] = \frac{1}{3}(1 + 2 e^{-
    3t}).$$ 
    Then we can compute 
    \begin{align*}  
    K^X(1,3|t) &= \int_{t_1+t_2+t_3 = t, t_i>0}dt_1dt_2\, e^{-t_1}\frac13\left(1+ 2e^{-3t_2}\right) e^{-t_3} \\
    &= \frac{1}{3} - \frac{e^{-t}}{2}+\frac{e^{-3t}}{6}, \\
    K^X(3,3|t) &= K_{X_2,Y}(3,3|t) + \int_{t_1+t_2+t_3 = t, t_i>0}dt_1dt_2\, e^{-t_1}\frac13\left(1+ 2e^{-3t_2}\right) e^{-t_3}  \\ 
     &= e^{-t} + \frac{1}{3} - \frac{e^{-t}}{2}+\frac{e^{-3t}}{6} = \frac{1}{3} + \frac{e^{-t}}{2}+\frac{e^{-3t}}{6}, \\
    K^X(1,2|t) &= \int_{t_1+t_2 = t}dt_1e^{-t_1}\frac13(1+ 2e^{-3t_2}) =\frac{1}{3} - \frac{e^{-3t}}{3}. \\
    K^X(2,2|t) &= \frac{1}{3}(1 + 2e^{-3t}) = \LL^{-1}[(\DD^{Y,X}_{m^2})^{-1}]
    \end{align*}
    which agrees with the direct computation of the matrix exponential as above and Remark \ref{rem: gluing formula on interface}.
\end{example}
\begin{example}
    For $p<q \in \ZZ$, Let us denote $[p;q]$ the line graph with vertices labeled $p,p+1,\ldots,q$. The heat kernel of the line graph $X=[0;N]$ with Dirichlet boundary condition on the two endpoints $\dd X=\{0,N\}$ is:\footnote{
    One can obtain this immediately from the eigenfunctions and eigenvalues of the graph Laplacian on the line graph. Alternatively, one can compute the inverse Laplace transform of the Green's function on the line graph, \cite[Example 3.7]{CKMW}.
    }
    \begin{equation*}
        K^{[0;N],\{0,N\}}(x,y|t)=\frac{2}{N}\sum_{j=1}^{N-1}\sin \frac{\pi j x}{N} \sin \frac{\pi j y}{N}\, e^{-\left(4\sin^2\frac{\pi j}{2N}\right) t}.
    \end{equation*}
If we represent the graph $[0;N]$ as a gluing of two line graphs $X_1=[0;N_1]$ and $X_2=[N_1;N=N_1+N_2]$ over the 1-vertex subgraph $Y=\{N_1\}$, 
\begin{figure}[h]
    \centering
    \includegraphics[scale=0.8]{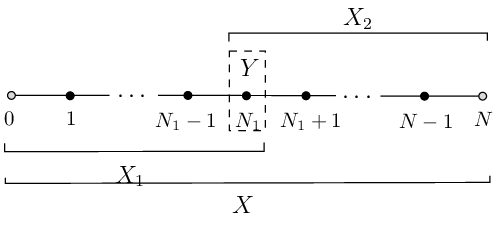}
    \caption{Line graph with Dirichlet condition on the end-vertices (drawn as hollow vertices), represented as a gluing of two subgraphs.}
    \label{fig: path N gluing}
\end{figure}
the gluing formula (\ref{eq: gluing heat kernel discrete}) takes the following form:
    \begin{multline*}
        K^{[0;N],\{0,N\}}(x,y;t)=
        K^{X_a,\dd X_a}(x,y;t) \delta_{ab}+\\+\int_{t_1+t_2+t_3=t,t_i>0} dt_1 dt_2 
        \underbrace{\left(\sum_{j=1}^{N_a-1}\frac{2}{N_a} \sin \frac{\pi j}{N_a} \sin \frac{\pi j\, \mr{dist}(x,Y)}{N_a}\, e^{-4 \sin^2\frac{\pi j}{2 N_a}\cdot t_1}\right)}_{\varepsilon^{Y,X_a}(t_1)} \\
        \underbrace{\left(\sum_{k=1}^{N-1}\frac{2}{N}\sin^2 \frac{\pi k N_1}{N} e^{-4\sin^2\frac{\pi k}{2N}\cdot t_2}\right)}_{\LL^{-1}[(\DD^{Y,X}_{m^2})^{-1}](t_2)} 
        \underbrace{\left(\sum_{l=1}^{N_b-1}\frac{2}{N_b} \sin \frac{\pi l}{N_b} \sin \frac{\pi l \,\mr{dist}(y,Y)}{N_b}\, e^{-4 \sin^2\frac{\pi l}{2 N_b}\cdot t_3}\right)}_{\varepsilon^{Y,X_b}(t_3)}.
    \end{multline*}
    Here we assume $x\in X_a, y\in X_b$ with $a,b\in \{1,2\}$.
\end{example}

\subsubsection{Discrete gluing formula II}
We can also prove a version of our ``second gluing formula,'' Theorem \ref{thm: gluing formula}. In the discrete case, we use another splitting of the Dirichlet-to-Neumann operator. Namely, using again the Schur complement formula we have 
$$ \DD^{Y,X}_{m^2} = D^{-1} = \wh{D} - \wh{C}\wh{A}^{-1}\wh{B}$$ 
and we have 
$$\wh{D} = (D^X)_{Y-Y} + m^2 - A^Y.$$
Here $(D^X)_{Y-Y}$ denotes the $Y-Y$ block of the degree matrix, i.e. it is a diagonal matrix whose entries are the valencies in $X$ of the vertices in $Y$.  
We therefore have a decomposition 
\begin{equation*}
    \DD^{Y,X}_{m^2} = \underbrace{(D^X)_{Y-Y} + m^2}_{=:\Lambda} - \underbrace{(A^Y + \wh{C}\wh{A}^{-1}\wh{B})}_{=:\DD'}.
\end{equation*} 
We then get a geometric progression for $\DD^{-1}_{Y,X}$:
\begin{equation*}
    (\DD^{Y,X}_{m^2})^{-1} = \Lambda^{-1}\sum_{k\geq 0}(\DD'\Lambda^{-1})^k = \Lambda^{-1} + \Lambda^{-1}\DD'\Lambda^{-1} + \ldots. 
\end{equation*}
The inverse Laplace transform of $\Lambda^{-1}$ is the diagonal matrix with entries $e^{-t\mr{val}_X(v)}$. We will give an explicit formula for the operator $\DD'$ and its inverse Laplace transform in terms of paths in Proposition \ref{prop: L inv D prime paths} ( Equation \eqref{eq: L inv D prime paths}) below. For now we remark that in the case where $X \setminus Y = (X_1 \setminus Y) \sqcup (X_2 \setminus Y)$ is a disconnected graph (i.e. there are no edges between $X_1 \setminus Y$ and $X_2 \setminus Y$), then the block decomposition \eqref{eq: block decomposition} becomes 
\begin{equation*}
     \Delta^X + m^2  = \left( 
    \begin{array}{c|c|c}
         \wh{A}_1=\Delta_{X_1,Y} + m^2& 0 & \wh{B}_1 \\ \hline
         0 & \wh{A}_2=\Delta_{X_2,Y} + m^2& \wh{B}_2 \\ \hline
         \wh{C}_1& \wh{C}_2 & \wh{D}
    \end{array}
    \right)
\end{equation*}
and therefore, we have 
\begin{equation*}
\DD' = A^Y + \wh{C}_1\wh{A}_1^{-1}\wh{B}_1 + \wh{C}_2\wh{A}_2^{-1}\wh{B}_2.
\end{equation*} 
This means that $\DD'$ can be expressed purely in terms of the two graphs $X_1,X_2$ and the interface $Y$. Similarly to Section \ref{rem: ext operator}, we can think of the second and the third term as combinatorial normal derivatives applied to both arguments of the Green's function with Dirichlet boundary conditions, compare this with the equation from the continuous case \eqref{eq: int kernel DN}. The inverse Laplace transform of $\DD'$ is 
\begin{equation*}
    \LL^{-1}[\DD'](t) = \delta(t)A^Y + \wh{C}_1K_{X_1,Y}(t)\wh{B}_1 + \wh{C}_2K_{X_2,Y}(t)\wh{B}_2.
\end{equation*}
(Notice that none of $A^Y$, $\wh{B_i}$, $\wh{C_i}$ depend on $m^2$.) We then have the following discrete version of Theorem \ref{thm: gluing formula}: 
\begin{proposition}\label{prop: gluing heat kernel discrete 2}
    Suppose that $X = X_1 \cup_Y X_2$. Let $x_1 \in X_{(a)}$, $x_2\in X_{(b)}$. Then 
    \begin{multline*}
        K^X(x_1,x_2|t) = \delta_{ab}K^{X_{(a)}}(x_1,x_2|t)+ \\ +\sum_{k\geq 0}\int_{\sum_{i=0}^{2k+2} t_i = t, t_i> 0}\prod_{i=0}^{2k+1}dt_i\sum_{u_0,\ldots u_k \in Y}
    \varepsilon^{Y,X_{(a)}}(x_1, u_0|t_0) \cdot \\ 
 \cdot \prod_{i=0}^{k-1}\left(  e^{-t_{2i+1}\mr{val}_X(u_{i})} \LL^{-1}[\DD'(u_i,u_{i+1})](t_{2i+2})\right)\cdot e^{-t_{2k+1}\mr{\val_X}(u_k)} \cdot \\
    \cdot(\varepsilon^{Y,X_{(b)}})^T(u_k, x_2|t_{2k+2}).
    \end{multline*}
\end{proposition}
 \subsection{Path sum formulae} 
 In this section, we will give alternative proofs of Proposition \ref{prop: gluing heat kernel discrete} and Proposition \ref{prop: gluing heat kernel discrete 2} (the two gluing formulas for the discrete heat kernel) that rely on path-counting arguments. 
 
 Let $P_X(u,v)$ be the set of paths from $u$ to $v$.\footnote{A \emph{path} from $u$ to $v$ is a sequence of vertices $(u = v_{0},v_{1},\ldots, v_{k} = v)$, with $v_{i} \neq v_{i+1}$. The length of this path is $k$. } In \cite{CKMW}, we prove the following path sum formula for the heat kernel. 
 \begin{proposition}[\cite{CKMW}]\label{prop: path sum heat kernel}
     We have 
     
     \begin{equation*}
         K(u,v|t) = \sum_{\gamma \in P_X(u,v)}W(\gamma,t),
     \end{equation*}
     where for a path $\gamma = (u=v_{0},v_{1},\ldots, v_{k}=v)$ we denote 
     \begin{equation*}
         W(\gamma,t) = \int_{\scriptsize
\begin{array}{c}
       t_0,\ldots,t_{k}> 0\\
     t_0+\cdots+t_{k}=t
\end{array}
} dt_1\cdots dt_{k}\, e^{-\sum_{j=0}^{k} t_j \val(v_{i})}.
     \end{equation*}
 \end{proposition}
For a path $\gamma = (v_{0},v_{1},\ldots, v_{k})$introduce the notations 
\begin{align*}
    \ogam &=  (v_{0},v_{1},\ldots, v_{k-1}),\\
    \ugam &= (v_{1},v_{1},\ldots, v_{k}),
\end{align*}
i.e., $\ogam$ is the $\gamma$ with the last vertex dropped and $\ugam$ is $\gamma$ with the first vertex dropped. The following lemma is elementary and we skip the proof: 
\begin{lemma}\label{lem: path weight}
    For any two composable paths $\gamma_1$ and $\gamma_2$ we have 
    \begin{equation*}
        W(\gamma_1 * \gamma_2) = \int_{\substack{t_1,t_2 >0 \\ t_1 + t_2 = t}} dt_1\,W(\ogam_1,t_1)W(\gamma_2,t_2) = \int_{\substack{t_1,t_2 >0\\ t_1 + t_2 = t}}dt_1 \, W(\gamma_1,t_1)W(\ugam_2,t_2).
    \end{equation*}
\end{lemma}
In particular, if $\gamma = \gamma_1 * \gamma_2 * \gamma_3$, then 
\begin{equation*}
    W(\gamma,t) = \int_{\substack{t_1,t_2,t_3 >0 \\ t_1 + t_2 + t_3= t}} dt_1dt_2\, W(\ogam_1,t_1)W(\gamma_2,t_2)W(\ugam_3,t_3) . 
\end{equation*}
Now, suppose that $Y\subset X$ is a full subgraph. Then we have an obvious decomposition $P_X(u,v) = P_{X\setminus Y}(u,v) \sqcup P_{X,Y}(u,v)$ of paths from $u$ to $v$ into does that contain no vertices in $Y$ and those that do.\footnote{In the case where $X \setminus Y = X_1 \sqcup X_2$ is disconnected, the set $P_{X \setminus Y}(u,v)$ is empty if $u \in X_1$ and $v \in X_2$ or vice versa.} Any $\gamma \in P_{X,Y}(u,v)$ can be decomposed as $\gamma = \gamma_1 * \gamma_2 * \gamma_3$, where $\gamma_1$ and $\gamma_3$ intersect $Y$ only in the ending (resp. starting) vertex, and $\gamma_2$ is an arbitrary path starting and ending on $Y$. This gives 
\begin{multline*}
    K(u,v|t) = \sum_{\gamma \in P_{X \setminus Y}(u,v)} W(\gamma,t) + \\ 
    + \sum_{y_1, y_2\in Y}\sum_{\gamma_1\in P'_{X,Y}(u,y_1)}\sum_{\gamma_2 \in P_X(y_1,y_2)}\sum_{\gamma_3\in P'_{X,Y}(y_2,v)}W(\gamma_1 * \gamma_2 * \gamma_3),
\end{multline*}
where we denote by $P'_{X,Y}(v,y)$ (resp. $P'_{X,Y}(y,v)$) the paths that end (resp. start) at a vertex $y \in Y$ and do not contain any other vertices in $Y$. Using Lemma \ref{lem: path weight}, we obtain: 
\begin{proposition}\label{prop: path sum gluing formula}
The heat kernel $K^X(t)$ on a graph admits the ``path sum gluing formula'' 
\begin{multline}
    K^X(u,v|t) = \sum_{\gamma \in P_{X \setminus Y}(u,v)} W(\gamma,t) + \\ 
    + \int_{\substack{t_1,t_2,t_3 >0 \\  t_1 + t_2 + t_3= t}} dt_1dt_2 \sum_{y_1, y_2\in Y}\sum_{\gamma_1\in P'_{X,Y}(u,y_1)}W(\ogam_1,t_1) \cdot \\ \cdot \sum_{\gamma_2 \in P_X(y_1,y_2)}W(\gamma_2,t_2)\sum_{\gamma_3\in P'_{X,Y}(y_2,v)}W(\ugam_3,t_3). \label{eq: path sum gluing formula}
\end{multline}
\end{proposition}
This is the discrete analog of the path integral formula \eqref{eq: path integral formula gluing}. \\ 
 The following proposition identifies the terms of the formula \eqref{eq: path sum gluing formula} with the integral kernels of the operators $K^{X,Y}(t)$, $\varepsilon^{Y,X}(t)$ and $\LL^{-1}[(\DD^{Y,X})^{-1}](t)$: 
\begin{proposition}\label{prop: path sum formulas heat kernels} We have
\begin{align*}
    K^{X,Y}(u,v|t) &= \sum_{\gamma \in P_{X \setminus Y}(u,v)} W(\gamma,t), \\
        \varepsilon^{Y,X}(u,y|t) &= \sum_{\gamma\in P'_{X,Y}(v_i,y)}W(\ogam,t), \\
        \LL^{-1}[(\DD^{Y,X})^{-1}](y_1,y_2|t) &= \sum_{\gamma_2 \in P_X(y_1,y_2)}W(\gamma,t).
    \end{align*}
\end{proposition}
\begin{proof}
    These formulae can be proved, exactly like Proposition \ref{prop: path sum formulas heat kernels}, by applying the inverse Laplace transform to the path sum expressions for the Green's function for $\Delta^{X,Y} + m^2$ 
    and the corresponding extension operator $E^{Y,X}_{m^2}$ which were derived in \cite{CKMW}.
\end{proof}
Propositions \ref{prop: path sum gluing formula} and \ref{prop: path sum formulas heat kernels} together imply the Proposition \ref{prop: gluing heat kernel discrete}.  \\
Finally, we remark that Proposition \ref{prop: gluing heat kernel discrete 2} can be proved by splitting paths that intersect $Y$ in a different way, namely, splitting at every vertex where they intersect $Y$.  Suppose a path $\gamma$ from $u$ to $v$ intersects $Y$ at $k+1$ vertices $y_0,\ldots,y_k$. For convenience, set $u = y_{-1}$ and $v = y_{k+1}$ and denote $\gamma_i$ the segment of $\gamma$ between $y_{i-1}$ and $y_i$. 
 \begin{figure}[h]
    \begin{tikzpicture}
\foreach \x in {0,1,...,7} 
{
\draw (\x,-0.5) -- (\x,7);
}
\foreach \y in {0,1,...,6} 
{\draw (-.5,\y) -- (8,\y);
}
\draw[line width = 2pt] (4,-.5) -- (4,7);
\draw[red, fill = red] (1,4) circle (3pt);
\draw[red, fill = red] (4,5) circle (3pt);
\draw[red, fill = red] (4,4) circle (3pt);
\draw[red, fill = red] (4,3) circle (3pt);
\draw[red, fill = red] (4,2) circle (3pt);
\draw[red, fill = red] (6,2) circle (3pt);

\draw[dashed] (4,3.25) ellipse (0.55cm and 4cm);
\draw[path] (1,4) -- (2,4) -- (2,5) -- (5,5) -- (5,6) -- (7,6) -- (7,5) -- (6,5) -- (6,4) -- (5,4) -- (3,4) -- (3,3) -- (4,3) -- (4,2) -- (6,2);
\node at (0.8,4.2) {$u$}; 
\node at (3.75,5.2) {$y_0$};
\node at (3.75,4.2) {$y_1$};
\node at (3.75,3.2) {$y_2$};
\node at (3.75,2.2) {$y_3$};
\node at (6.2,1.7) {$v$};
\node at (1.8,5.2) {$\gamma_0$};
\node at (5.25,5.5) {$\gamma_1$};
\node at (2.75,3.5) {$\gamma_2$};
\node at (4.25,2.5) {$\gamma_3$};
\node at (5.5,1.75) {$\gamma_4$};

\node at (1,-1.3) {$X_1$};
\node at (4,-1.3) {$Y$};
\node at (6,-1.3) {$X_2$};
\end{tikzpicture} 

\caption{Decomposing a path $\gamma$ into components $\gamma_i$.  }
        \label{fig: gamma decomposition}
        \end{figure}
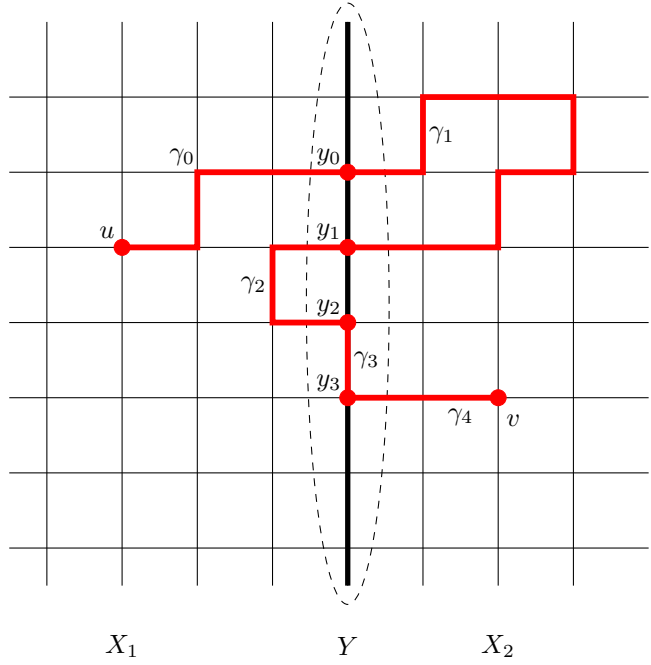
        Then, iterating Lemma \ref{lem: path weight}, we get
\begin{multline}\label{W via cutting}
    W(\gamma,t) =\int_{\sum_{i=0}^{2k+2} t_i = t, t_i> 0}\prod_{i=0}^{2k+1}dt_i \\ W(\ogam_0,t_0)W(y_0,t_1)W(\ougam_1,t_2)\cdots W(\ougam_k,t_{2k})W(y_k,t_{2k+1})W(\ugam_{k+1},t_{2k+2}) .
\end{multline}
Here we abuse notation by denoting a length 0 path by its only vertex, i.e., $W(y_i,t) = e^{-t\val_X(y_i)}$. Also, notice that the paths $\ougam_i$ can be empty, in the case where $\gamma_i = (y_{i-1},y_i)$ (i.e. the segment $\gamma_i$ is a single edge contained in $Y$, see Figure \ref{fig: gamma decomposition}). In this case we set $W(\ougam_i,t) = \delta(t)$. 

Now, Proposition \ref{prop: gluing heat kernel discrete 2} follows from Propositions \ref{prop: path sum gluing formula} and \ref{prop: path sum formulas heat kernels}, expansion (\ref{W via cutting}), and the following path sum formula for $\LL^{-1}[\DD'](t)$: 
\begin{proposition}\label{prop: L inv D prime paths}
    For two vertices $y_1,y_2 \in Y$ denote by $P''_{X,Y}(y_1,y_2)$ all paths from $y_1$ to $y_2$ in $X$ that do not intersect $Y$ except in $y_1$ and $y_2$. Then 
    \begin{equation}
        \LL^{-1}[\DD'](y_1,y_2|t) = \sum_{\gamma \in P''_{X,Y}(y_1,y_2)}W(\ougam,t).
        \label{eq: L inv D prime paths}
    \end{equation}
\end{proposition}
\begin{proof}
    This follows from the path sum formula for $\DD'$ that was proved in \cite{CKMW} after applying the inverse Laplace transform. 
\end{proof}

\end{document}